\documentclass[journal,10pt]{IEEEtran}
\usepackage{graphicx, epstopdf}
\usepackage{amssymb, amsmath}
\usepackage{algorithm}
\usepackage{algorithmic}
\usepackage{subeqnarray}
\usepackage{cases}
\usepackage{commath}
\usepackage{setspace}
\usepackage{acronym}
\usepackage{mathrsfs}
\usepackage{ifthen}
\usepackage{textcomp}
\usepackage{float}
\usepackage{subfigure}
\usepackage[table]{xcolor}
\usepackage{color,soul}
\usepackage{cite}
\usepackage{setspace}
\usepackage{mathtools}
\usepackage{bm}
\usepackage{url}
\usepackage{xcolor}
\usepackage{empheq}
\usepackage{tikz}
\usepackage{subfigure}
\usepackage{stfloats,amsthm}
\usepackage{multirow}
\usepackage{multicol}
\usepackage{arydshln}
\usepackage{footnote}
\usepackage{booktabs}
\usepackage{threeparttable}
\makesavenoteenv{table}
\newtheorem{theorem}{Theorem}
\newtheorem{lemma}{Lemma}

\DeclareMathSizes{12}{12}{7.7}{5.5}
\usepackage[bb=boondox,bbscaled=1]{mathalfa}

\def\A{{\bf A}}
\def\a{{\bf a}}

\def\C{{\bf C}}

\def\G{{\bf G}}

\def\I{{\bf I}}

\def\Q{{\bf Q}}

\def\R{{\bf R}}

\def\S{{\bf S}}

\def\U{{\bf U}}
\def\u{{\bf u}}
\def\V{{\bf V}}
\def\v{{\bf v}}
\def\W{{\bf W}}

\def\X{{\bf X}}
\def\x{{\bf x}}
\def\Y{{\bf Y}}

\def\0{{\bf 0}}
\def\1{{\bf 1}}

\def\IM{{\mathcal I}}

\def\NM{{\mathcal N}}
\def\OM{{\mathcal O}}

\def\CB{{\mathbb C}}

\def\Si{\mbox{\boldmath$\Sigma$\unboldmath}}

\def\Pii{\mbox{\boldmath$\Pi$\unboldmath}}

\def\argmin{\mathop{\rm argmin}}

\def\orth{\mathsf{orth}}

\def\rk{\mathrm{rank}}

\begin{document}

\title{Ultra-Fast Accurate AoA Estimation via Automotive Massive-MIMO Radar}

\author{Bin~Li$^{1,2}$,
        Shusen~Wang$^{3}$,
        Jun~Zhang$^{2}$,
        Xianbin~Cao$^{4}$,
        Chenglin~Zhao$^{2}$

\thanks{\emph{1} School of Information and Electronics, Beijing Institute of Technology, Beijing, 100081, China.}
\thanks{\emph{2} School of Information and Communication Engineering,
Beijing University of Posts and Telecommunications, Beijing, 100876, China.}
\thanks{\emph{3} Department of Computer Science, Stevens Institute of Technology, Hoboken, NJ 07030, USA.}
\thanks{\emph{4} School of Electronic and Information Engineering, Beihang University, Beijing, 100191, China.}
}

\maketitle

\begin{abstract}
Massive multiple-input multiple-output (MIMO) radar, enabled by millimeter-wave virtual MIMO techniques, provides great promises to the high-resolution automotive sensing and target detection in unmanned ground/aerial vehicles (UGA/UAV).
As a long-established problem, however, existing subspace methods suffer from either high complexity or low accuracy.
In this work, we propose two efficient methods, to accomplish fast subspace computation and accurate angle of arrival (AoA) acquisition.
By leveraging randomized low-rank approximation, our fast multiple signal classification (MUSIC) methods, relying on random sampling and projection techniques, substantially accelerate the subspace estimation by orders of magnitude.
Moreover, we establish the theoretical bounds of our proposed methods, which ensure the accuracy of the approximated pseudo-spectrum.
As demonstrated, the pseudo-spectrum acquired by our fast-MUSIC would be highly precise; and the estimated AoA is almost as accurate as standard MUSIC.
In contrast, our new methods are tremendously faster than standard MUSIC.
Thus, our fast-MUSIC enables the high-resolution real-time environmental sensing with massive MIMO radars, which has great potential in the emerging unmanned systems.
\end{abstract}


\IEEEpeerreviewmaketitle

\section{Introduction}

Unmanned aircrafts and automotive vehicles are receiving more and more attention from industry and academia \cite{jones2001keeping,Guizzo2011Google}, owing to its great potential in the widespread applications~\cite{Blasch2006Unmanned}.
The success of unmanned systems critically depends on the fusion of Global Positioning System (GPS), automotive radar, and other environment sensors (e.g. lidar, ultrasound, and camera) \cite{Alonzo2004Toward,murad2013requirements}.
In comparison to others techniques, millimeter-wave (mm-wave) automotive radar is extremely attractive to future unmanned systems, due to its two inherent merits.
First, it is immune to adverse environment such as dust, fog, and smoke;
and second, it is also robust to dazzling and no light conditions \cite{waldschmidt2014future}.
Furthermore, recent advancement in mm-wave semiconductor circuit (e.g. 24 and 77GHz) \cite{hasch2012millimeter} makes it possible to deploy the large-scale arrays economically onto unmanned systems.
As such, assisted by virtual Multiple-Input Multiple-Output (MIMO) or co-located MIMO techniques \cite{li2007mimo}, hundreds of equivalent receiving channels are made available to achieve the high-resolution environment sensing in dynamical automotive scenarios \cite{bilik2018automotive,2019The}.

In practice, unmanned vehicles use three types of mm-wave radars \cite{hasch2012millimeter}:
(1) long-range radar (LRR) for automotive cruise control, in the range $\sim$200 m and with field of views (FOVs) of azimuth $\pm 9^\circ$; (2) medium-range radar (MRR) for cross-traffic alert ($\sim$70 m, FOVs $\pm 15^\circ$); and (3) short-range radar (SRR) for park assist ($\sim$10 m, FOVs $\pm 75^\circ$);
as illustrated in Figure \ref{fig:radar}(a).
To meet diverse requirements, various modulation waveforms are also designed for automotive radars\cite{patole2017automotive,2019High}.
Popular solutions include frequency-modulated continuous-waveform (FMCW) \cite{stove1992linear,Lin2011A,Babur2013Nearly} and pulsed continuous wave (CW) radar, etc.
Among these, FMCW sweeps a broad Radio Frequency (RF) bandwidth (GHz) while maintaining a small Intermediate Frequency (IF) bandwidth (MHz), which permits the high-resolution sensing via low-cost circuit \cite{bilik2018automotive}, thus providing great potential to the practical deployment.

Despite the great advances in hardware integration and system design in automotive massive-MIMO radars~\cite{patole2017automotive,Kok2005Signal,2020MIMO}, high-resolution and low-complexity signal processing remains one substantial challenge.
For massive-MIMO radars, the high-resolution environment sensing requires a large number of channels \cite{bilik2018automotive}, which incurs a high computation cost.
Meanwhile, computational data processing causes an intolerable latency, which may lead to disastrous consequences.
Thus, the deployment of massive MIMO automotive radar calls for high-resolution yet real-time processing algorithms.

Theoretically, the maximum-likelihood (ML) method achieves the optimal accuracy \cite{Wenger2005Automotive}, with an exhaustive search in two-dimensional large space (i.e. range and AoA).
Provided a large number of samples, subspace methods, represented by Multiple Signal Classification (MUSIC) \cite{schmidt1986multiple} and Estimation of Signal Parameters via Rational Invariance Techniques (ESPRIT) \cite{roy1989esprit}, are able to attain the near-optimal accuracy \cite{ krim1996two}, relying on signal or noise subspaces extracted by singular value decomposition (SVD) on a covariance matrix.
Despite the high-resolution sensing, the high complexity, i.e. $\mathcal{O}(M^3)$ for 1-D estimation and $\mathcal{O}(N^3M^3)$ for 2-D estimation ($M$ is the number of antennas and $N$ is the number of snapshots), as well as the intolerable latency seriously limit the practical use~\cite{oh2015low}, especially in massive-MIMO radars ($M>200$)\footnote{For example, when the number of elements is 1000 as in emerging MIMO automotive radars, the processing latency would be around 500 ms even on high-performance CPUs.}.
In practice, the dimension of interests of MUSIC can be reduced by pre-estimation~\cite{zoltowski1993beamspace}, which is still computational (e.g. an extended matrix in 2-D MUSIC is very large).

To simplify such subspace methods, a block-Lanczos scheme was applied \cite{simon1984the}, which focuses only on the first $K$ singular vectors and thus reduces the complexity to $\OM (p_MKM^2)$, whereby $p_M \sim \mathcal{O}(\log M)$ is the number of blocks in the Krylov subspace and $K$ is the number of targets.
The so-called beamspace method, i.e. transforming element space to beam-space via a beamforming weight vector \cite{Mathews1994Eigenstructure,2018Lin}, can be used to simplify the subspace computation.
Although it attains a promising performance in uniform circular array (UCA) \cite{2018Lin}, it may seriously degrade the AoA accuracy in the case of uniform linear array (ULA), as for automotive MIMO radars.
Compressive sensing (CS) can be also applied to AoA estimation, which may be computational for massive MIMO radar \cite{2013Spatial,Stoeckle2015DoA,2020MIMO,2015Distributed}.
In ref. \cite{Benesty2005A}, the complex SVD was replaced with matrix inverse~\cite{oh2015low}.
Although the latter may be realized more efficiently, it has the same complexity $\OM (M^3)$.
When the signal-to-noise-ratio (SNR) is low, the estimated noise subspace would be even inaccurate, making the approximated pseudo-spectrum largely deviated.
In ref. \cite{marcos1995the}, the Propagator scheme was designed.
By introducing a Propagator operator, it simplifies the complexity \cite{Marcos1994Performances}, yet scarifying the spatial resolution and the estimation reliability (e.g. in the range of $[70,~90]$ degree)~\cite{Tayem2005L}.
As one popular solution, unknown targets can be quickly estimated via FFT~\cite{Wenger2005Automotive,engels2017advances,Garcia2018TI}.
Due to its low complexity, $\mathcal{O}(M\text{log}_2M)$, FFT has been widely used in automotive radars \cite{Ramasubramanian2017TI,Garcia2018TI,2019The}, which, however, is insufficient for the high-resolution sensing and high-quality point cloud acquisition~\cite{patole2017automotive}.

In this work, we break this major bottleneck in massive MIMO automotive radar.
That is to say, we aim to accomplish real-time automotive sensing at a scalable computational complexity, but achieve the high-resolution target estimation as the standard MUSIC.
To achieve this, we apply randomized low-rank approximation (RLAR) to a covariance matrix, and develop two fast-MUSIC methods to approximately compute the signal subspace of it -- the general stumbling block in all subspace-based methods.
As such, we are able to acquire the highly accurate AoA estimation at a linear complexity -- $\mathcal{O}(K^2 M)$.
Our fast-MUSIC would boost the widespread use of massive MIMO radars in the emerging unmanned systems.

In summary, we offers the following contributions.
\begin{itemize}
	\item
    We introduce a novel concept, i.e. \emph{approximated} rather than traditional \emph{exact} matrix computation, to subspace methods.
    We first approximate a large covariance $\S$ via three small matrices, in one special form $\S\simeq \C\W\C^H$.
    Here, the matrix sketch $\C \in \CB^{M \times p}$ ($K <  p \ll M$) is abstracted by uniform column-sampling on $\S$; while $\W\in \CB^{p \times p}$ is a weight matrix minimizing the approximation residue.
    Relying on such a randomized low-rank approximation, our fast-MUSIC reduces the time complexity of subspace computation to $\OM (p^2 M )$.
    \item
	We show our fast-MUSIC method enables the highly accurate pseudo-spectrum approximation and the precise AoA estimation.
    As the theoretical analysis suggests, when the signal subspace is approximately computed as in our fast-MUSIC, the estimated pseudo-spectrum $\sqrt{\tilde{P}_{\textrm{music}}(\theta)}$ has a relative error
	\[
	\OM \Big( \tfrac{ \sigma_{K+1} (\S) }{ \sigma_K (\S) } \, \sqrt{\tfrac{M^2}{p}} \Big),
	\]
	compared to the standard MUSIC pseudo-spectrum $\sqrt{P_{\textrm{music}}(\theta)}$, when the user-specify parameter meets $p \thicksim \OM(K \log K)$.
	Here, $\sigma_k (\S)$ is the $k$-th largest singular value of $\S$.
	If the signal-to-noise ratio (SNR) is relatively high, the spectral gap $\tfrac{ \sigma_{K+1} (\S) }{ \sigma_K (\S) }$ is very small, making our estimation highly accurate.
	\item
	To further improve the approximation accuracy, we resort to another random projection technique, by iteratively finding a good orthogonal projection matrix $\V\in \CB^{M \times p}$.
    In contrast to a direct uniform sampling, we obtain an improved sketch $\C=\S\V\in \CB^{M \times p}$ and a low-rank approximation $\S\simeq \C\W\C^H$.
    By incorporating more information into $\C$, the accuracy of the approximated pseudo-spectrum is further enhanced.
    \item	
    We establish the theoretical bound of our refined fast-MUSIC method.
    As shown, after $t$ iterations, the approximated pseudo-spectrum $\sqrt{\tilde{P}_{\textrm{music}}(\theta)}$ has a relative error
	\[
	\OM \left( \big(  \tfrac{ \sigma_{K+1} (\S) }{ \sigma_K (\S) } \big)^{t+1} \, \sqrt{M^2 K} \right),
	\]
	compared to the standard MUSIC pseudo-spectrum $\sqrt{P_{\textrm{music}}(\theta)}$.
	In this case, even if the SNR is \emph{not} very high, several power iterations (e.g. $t=2$) suffices for the near-optimal estimation.
    \item
    We evaluate our fast-MUSIC methods with comprehensive simulations, and compare them with its subspace counterparts (e.g. MUSIC, Lanczos, Propagator, etc).
    Our numerical studies validate the efficiency and the accuracy of fast-MUSIC, which provides the great potential to high-resolution and real-time massive-MIMO radars in the emerging automotive applications.
\end{itemize}

The remaining of this work is structured as follows.
In Section II, we briefly introduce the system model and the subspace methods for automotive sensing.
In Section III, we present two fast-MUSIC methods inspired by randomized matrix sketching.
In Section IV, we derive some theoretical bounds of interests.
In Section V, the numerical simulations are provided.
Finally, we conclude this work in Section VI.

The related notations are summarized as follows:
$\A$ denotes an matrix with a rank $r = \rk (\A)$;
its SVD is $ \A \: = \: \U \Si \V^H \: = \: \sum_{i=1}^r \sigma_i (\A) \, \u_i \v_i^H ,$
where $\sigma_i (\A)$ ($> 0$) is the $i$-th singular value;
$\u_i$ and $\v_i$ are the $i$-th left and right singular vector.
The Moore-Penrose pseudo-inverse of $\A$ is $\A^\dagger$.
The Frobenius norm of $\A$ is $\| \A \|_F \:  = \: \sqrt{\text{tr}(\A^H \A)}.$
The best rank-$k$ ($k < r$) approximation to $\A$ is $ \A_k \: \triangleq \: \sum_{i=1}^k \sigma_i (\A) \, \u_i \v_i^H . $

\section{Subspace Methods in MIMO Radar}

In this section, we briefly introduce FMCW MIMO radar of autonomous systems.
On this basis, some popular subspace methods for high-resolution estimation are analysed.

\subsection{System Model -- FMCW Radar}
A massive-MIMO radar system (e.g. virtual MIMO) consists of $M$ receiving antennas, i.e. typically the ULA.
For each element, the emitted FMCW waveform within a symbol duration $T_{\text{sym}}$ reads \cite{stove1992linear,Garcia2018TI}:
\begin{align} \label{eq:def_st}
s(t)= \text{exp} \left[ j \big(w_s t + \frac{\mu}{2}t^2 \big) \right], ~~~0\leq t < T_{\text{sym}},
\end{align}
where $w_s$ is the initial frequency and $\mu$ is the changing rate of instantaneous frequency of the chirp signal.
Given the bandwidth of FMCW signals $w_B$, as well as the emission duration, $T_{\text{sym}}$, the changing rate is $\mu=w_B/T_{\text{sym}}$.

After the signal calibration (i.e. mitigating non-ideal antenna effects based on the measured response), the received signal at the $m$-th antenna reads:
\begin{align*}
y_m(t) =  \sum_{k=0}^{K-1}\alpha_{k,m} s (t-\tau_k)\times \text{exp} \left( j \frac{2\pi}{\lambda_s} m d  \text{sin}\theta_k \right) + n_m(t).
\end{align*}
Here, $r_k(t)\triangleq\alpha_{k,m} s (t-\tau_k)$ is the signal reflected from the $k$-th target;
$\tau_k$ and $\theta_k$ are times of arrival and angles of arrival;
$\alpha_{k,m}$ is the complex channel gain between the $k$-th target and the $m$-th element, and in the far-filed cases\footnote{When the near-field effect occur, the subspace method would be still applicable, e.g. with the proper modification \cite{2012Efficient}. In such cases, our proposed methods would be applied directly to reduce the time complexity of subspace computation.} we have $\alpha_{k,m}=\alpha_k$;
$\lambda_s$ denotes the wavelength of carrier;
$d$ is the spacing distance of two adjacent elements (we assume $d=\lambda_s/2$);
$n_m(t) =  \frac{1}{\sqrt{2}}[\mathcal{N}(0,\sigma_{n}^2) + j \mathcal{N}(0,\sigma_{n}^2)]$ denotes the complex additive Gaussian noise with a variance $\sigma_{n}^2$.

A significant advantage of FMCW MIMO radar is that it is capable of demodulating signals via one simple mixer \cite{stove1992linear}, which greatly facilitates the low-cost RF implementation.
To be specific, at the reception-end the mixer-based de-chirping process is used to extract the beat signal $\tilde{y}_m(t)$, i.e.
\begin{align*}
\check{y}_m(t)= s(t)\times \left[\sum_{k=0}^{K-1}r_k (t) \text{exp} \left( j\frac{2\pi}{\lambda_s}m d  \text{sin}\theta_k \right) + n_m(t)\right].
\end{align*}
By substituting $s(t)$ into eq. \eqref{eq:def_st}, and after a low-pass filter with an impulse response $h(t)$, the target signal can be obtained, which is a composition of target-modulated signals:
\begin{align}
\tilde{y}_m(t)=  \sum_{k=0}^{K-1}\alpha_k   e^{ j\left( \mu \tau_kt + \omega_s\tau_k-\frac{\mu}{2}\tau_m^2 \right)} \cdot e^{ j\left( \pi k \text{sin}\theta_k \right)} + \tilde{n}_m(t). \nonumber
\end{align}
Here, $\tilde{n}_m(t) = h(t) \otimes [n_m(t)s(t)] $ is the residual noise after the de-chirping procedure; $\otimes$ denotes the convolution process.
Then, the analog to digital convertor (ADC) with a sampling frequency $f_s=1/T_s$ outputs the discrete signal:
\begin{equation*}
\tilde{y}_m (n) \: = \: \tilde{y}_m(nT_s) ,
\end{equation*}
for $m=0$ to $M-1$ and $n=0$ to $N-1$, with $N\triangleq T_{\text{sym}}/T_s$.
On this basis, the estimation of AoA, range and velocity of targets can be estimated from the discrete signal $\tilde{y}_m (n)$.
\begin{figure}[!t]
\centering
\includegraphics[width=8cm]{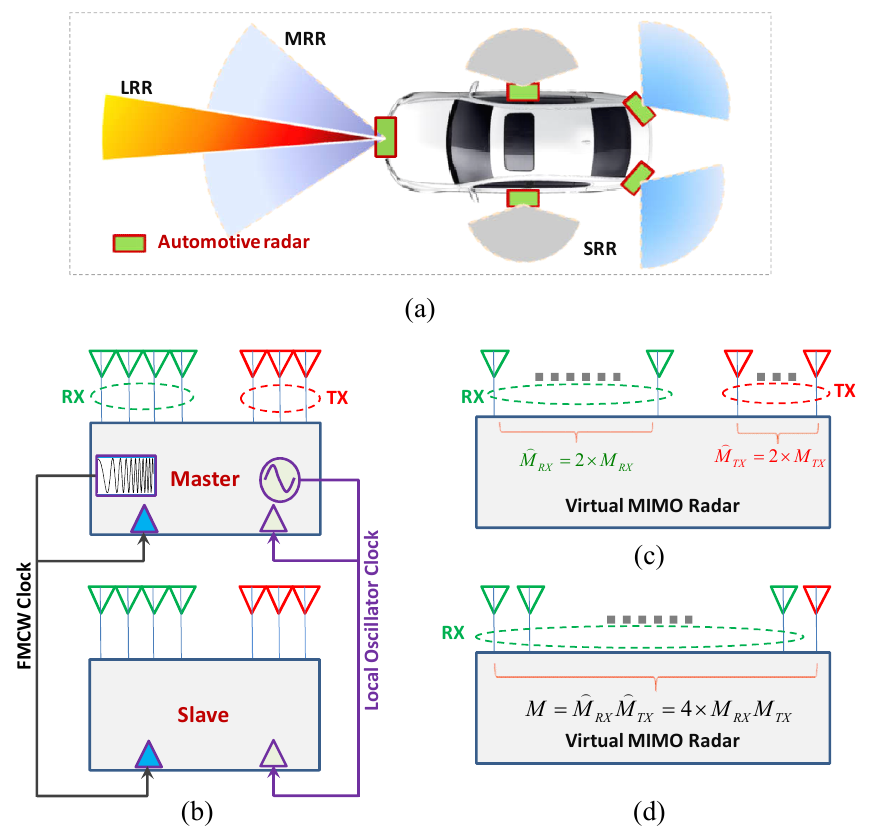}
\caption{\textbf{Massive MIMO automotive radar with cascaded sub-arrays}. (a) Typical radars for autonomous driving systems, including LRR, MRR and SRR. (b) A schematic framework on cascaded MIMO radar, whereby the master generates the local oscillator signal and controls the slave via two synchronized clocks, i.e. FMCW clock and LO clock. Here, each co-located MIMO radar is equipped with $M_{TX}=3$ transmitting antennas and $M_{RX}=4$ receiving antennas. (c) The equivalent MIMO radar system of two cascaded sub-systems. (d) The cascaded-based massive MIMO radar system with the equivalent channels of $M=4\times M_{RX}M_{TX}$.}
\label{fig:radar}
\end{figure}

Combined with the massive MIMO techniques, the resolution of automotive radars can be effectively improved.
In practice, this can be implemented via virtual MIMO techniques \cite{li2007mimo,bilik2018automotive,2019The,2019High}.
As in Figure 1-b, each subsystem involves 3 TX and 4 RX elements -- the receiving channels is then 12.
When two sub-systems are cascaded ($N_c=2$), the equivalent receiving channels would be $(3\times 2)\times(4\times2)=48$, see Figure 1-c.
In the general case, for the $N_c$ cascading system, the equivalent receiving channels would be $12 N_c^2$, as illustrated in Figure 1-d.
Thus, hundreds of elements can be integrated to enhance the spatial resolution \cite{bilik2018automotive,Li2020Fast}.
Moreover, this leads to one compact MIMO radar (e.g. $\sim$10 cm array), thus alleviating the near-field effects to some extent.


\subsection{Subspace Detection \& Estimation} \label{sec:prior_works}

For automotive sensing, unknown AoA (i.e. $\theta_k$) can be estimated by various methods \cite{patole2017automotive,2019High,Kok2005Signal}.
When the high resolution estimation was emphasised as in automotive sensing, subspace methods tend to be more attractive \cite{2020MIMO,patole2017automotive}.
Popular subspace methods, such as MUSIC \cite{schmidt1986multiple} and ESPRIT \cite{roy1989esprit}, all start from a covariance matrix, i.e.
\begin{align}
\S = 1/N \times \Y \Y^H,
\end{align}whereby the signal matrix $\textbf{Y} \in \mathbb{C}^{M \times N}$ is arranged as:
\begin{align}
\textbf{Y} \triangleq \begin{bmatrix}
\tilde{y}_0(0) & \tilde{y}_0(1) & \cdots & \tilde{y}_0(N-1)\\
\tilde{y}_1(0) & \tilde{y}_1(1) & \cdots & \tilde{y}_1(N-1)\\
\vdots & \vdots & \ddots & \vdots \\
\tilde{y}_{M-1}(0) & \tilde{y}_{M-1}(1) & \cdots & \tilde{y}_{M-1}(N-1)
\end{bmatrix}.
\end{align}

There are two important things to be noted for the above covariance matrix.
First, the matrix multiplication in computing $\S$ incurs also a high complexity, which yet can be efficiently calculated in parallel.
Second, for massive-MIMO radars (with a large $M$), the required snapshot number $N$ will be very large, in order to estimate an unbiased covariance.
For a small length, the standard MUSIC pseudo-spectrum may be biased \cite{stoica1989music}, e.g. due to the so-called subspace swapping problem \cite{fortunati2019scaling}.
Such theoretical limitations of the standard MUSIC may be overcame by some post-processing algorithms, e.g. correcting the estimated pseudo-spectrum via well-designed weights \cite{vallet2015performance}.
Recently, a single-snapshot MUSIC was developed \cite{Liao2014MUSIC}.
As shown, even in the extreme case $N=1$, the high-resolution AoA estimation can be attained.


Note that, the emphasis of this work is on the subspace computation which is one prerequisite to the accurate estimation of AoA.
Assume there are $K$ unknown target points, we first compute SVD of the large covariance $\S \in \mathbb{C}^{M \times M}$, i.e.
\begin{equation}\label{eq:r}
\S
\: = \: \U \Si \V^H
\: = \: \U_K \Si_K \V_K^H + \epsilon_n^2 \U_{-K} \V_{-K}^H,
\end{equation}
where $\U_K = [\u_1 , \cdots , \u_K] \in \mathbb{C}^{M \times K}$ corresponds to the signal subspace,
and $\U_{-K}= [\u_{K+1} , \cdots , \u_M] \in \mathbb{C}^{M \times (M-K)}$ corresponds to the noise subspace.
For a Hermitian matrix $\S^H=\S$, we further have $\V=\U$.
As common, we assume $\sigma_1 (\S) > \sigma_2 (\S) > \cdots > \sigma_M (\S)$, and $\sigma_m (\S) = \epsilon_n^2~~(m>K)$.

Then, we uniformly partition the spatial angle range $[0, \pi]$ into a grid of $L$ values, $\{\theta_0, \theta_1, \cdots , \theta_{L-1}\}$.
For each angle $\theta_l$, we define the $M$-dimensional steering vector as:
\begin{equation*}
\a (\theta)
\: = \: \big[
a_0 (\theta), \, a_1 (\theta),\, \cdots , \, a_{M-1} (\theta)
\big]^H
\end{equation*}
where
\begin{equation*}
a_m (\theta )
\: \triangleq \:
\exp \big[\tfrac{j \, 2\pi \, d \, m \cdot \text{sin}(\theta)}{\lambda}  \big],
~~m = 0, 1, \cdots, M-1.
\end{equation*}

Finally, we are able to evaluate the standard MUSIC pseudo-spectrum, based on the computed subspaces, i.e.
\begin{eqnarray} \label{eq:p_music}
P_{\text{music}}(\theta)
= \tfrac{1}{\a (\theta)^H  \U_{-K} \U_{-K}^H \a(\theta)}
= \tfrac{1}{\a (\theta)^H (\I_M - \U_{K} \U_{K}^H )\a(\theta)}.
\end{eqnarray}
And then, each peak in the pseudo-spectrum $P_{\text{music}}(\theta)$ indicates one target point.
So, unknown AoAs, $\theta_k~(k=0,1,\cdots,K-1)$, are attained by
\begin{eqnarray}
\hat{\pmb{\theta}} =  \big\{\theta_k,~ \theta_k=\texttt{peak}\{P_{\text{music}}(\theta) \} \big\}.
\end{eqnarray}

In practice, the pseudo-spectrum $P_{\text{music}}(\theta)$ needs to be computed for every $\theta \in \{ \theta_0, \cdots , \theta_{L-1} \}$.
The computation of the whole pseudo-spectrum $P_{\text{music}}(\theta)$ requires a time complexity $\OM (MKL)$.
Fortunately, $P_{\text{music}}(\theta)$ can be evaluated in a highly efficient manner~\cite{roy1989esprit,Li2021Fast}, which is hence out of the scope of this work.
In the following, we are interested in the fast computation of a signal subspace $\U_K$.
As seen, the involved complexity mainly comes from the decomposition of large covariance $\S$, which is measured by $\OM (M^3)$; and more importantly, there is no parallel algorithm to accelerate it.
Thus, despite its high-resolution in automotive sensing, the processing latency of subspace methods would be intolerable for massive-MIMO radars, e.g. $>100$ ms for $M>500$.

\begin{table}[]
\caption{Time complexities of classical DoA estimation methods}
{\centering
\begin{tabular}{@{}llll@{}}
\toprule
\textbf{Methods}                             & \textbf{Complexity}                                     & \textbf{Resolution} & \textbf{Conditions}   \\ \midrule
MUSIC              & $\mathcal{O}(M^3)$  & high       & all          \\
Lanczos               &  $\mathcal{O}(p_MKM^2)$  & high       & $K \leq M $           \\
Matrix-inverse          & $\mathcal{O}(M^3)$  & high       & $\epsilon_n \rightarrow$ 0    \\
Propagator             & $\mathcal{O}(KMN)$  & medium     & $|\theta|\leq$70 deg \\
FFT Method                              & $\mathcal{O}(M\text{log}M)$  & low        & all          \\
\emph{Fast MUSIC 1}                              & $\mathcal{O}(Mp^2)$  & high        & $K \ll M $         \\
\emph{Fast MUSIC 2}                              & $\mathcal{O}(tpM^2)$  & high        & $K \ll M $         \\\bottomrule
\end{tabular}
\begin{tablenotes}
\footnotesize
\item 1. Here, $p$ ($K \leq p \ll M$) is an over-sampling parameter;
      $t$ ($t \geq 1$) is the number of iterations in updating the projection matrix.
\item 2. Note that, the real complexity of our fast method 2 is largely lower than $\mathcal{O}(tpM^2)$, as the polynimial complexity on $M$ comes from the matrix multiplication which is yet more efficient than SVD or matrix inverse.
\item 3. $p_M \sim \mathcal{O}(\log M)$ is the number of blocks in constructing a Krylov subspace in the Lanczos method.
\end{tablenotes}
}
\end{table}

\subsection{Considerations on Subspace Methods}

As mentioned, there are many works which target at reducing the time complexity of subspace methods, which have different advantages and disadvantages, as summarized in Table I.
From this comparative result, one long-standing challenge in MIMO radar signal processing is that the high resolution and the low complexity have the inherent conflict.
Meanwhile, most existing algorithms manage to \emph{compromise} such two contradictory performance aspects.
\section{Fast MUSIC Methods}
In this section, we develop two highly efficient methods for subspace estimation.
Note that, various low-rank approximation techniques, e.g. SVD, have been widely applied to the subspace-based AoA estimation \cite{2020Transform,2019DOA,2020Target}.
In contrast to previous works mainly balancing two contradictory aspects, here we design two fast-MUSIC algorithms which permits the highly accurate AoA estimation at a linearly scalable complexity.
In principle, our new methods are inspired by the random sampling/projection techniques \cite{wang2019scalable,drineas2005nystrom,woodruff2014sketching,2021Random}.
Despite their great interests in linear algebra or scientific computing, such randomized sketch methods have been rarely considered in the field of MIMO signal processing \cite{Li2020Fast,Li2020Randomized}.
In sharp contrast to our previous work \cite{Li2020Fast} that adopts a less efficient sketch form and enjoys no theoretical bound, this would be the first attempt to develop ultra-fast subspace methods with the close-form error bound in the approximated pseudo-spectrum.

\subsection{Fast MUSIC Method -- 1}

In our fast-MUSIC algorithm, we firstly approximate the large covariance $\S \in \mathbb{C}^{M \times M}$ with three small matrices (referred as matrix sketches), in the special form:
$\S\simeq \C\W\C^H$ (as in Figure 2).
Meanwhile, in order to minimize the computational complexity, we resort to a simple random sampling technique to abstract the involved matrix sketch $\C \in \CB^{M \times p}$ ($K \leq p \ll M$) from $\S$.

\begin{figure}[!t]
\centering
\includegraphics[width=8.6cm]{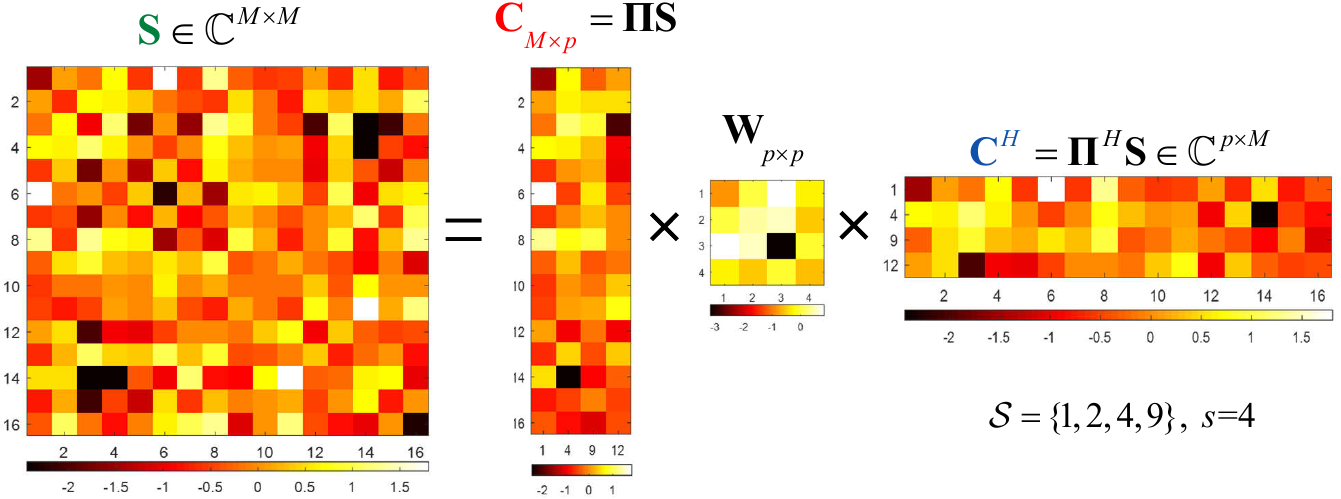}
\caption{Illustration of randomized low-rank approximation of a large covariance matrix $\S \in \mathbb{R}^{N \times N}$, with rank$(\S)=4$. Here, the randomly selected column indexes set is $\mathcal{I}=\{1,4,2,9\}$, with $|\mathcal{I}|=4$. Thus, the low-dimensional sketch $\C$ consists of $p=|\mathcal{I}|$ columns of $\S$, while $\C^H$ involves $4$ rows. The resulting approximation error is less than $10^{-12}$. }
\label{fig:2}
\end{figure}
To be specific, we first configure one over-sampling parameter\footnote{Here, the term over-sampling indicates that the sampling parameter $p$ should be always greater than the number of target points $K$, i.e. $p> K$.}, which is denoted as $p$; and thus we uniformly sample $p$ items from $\{0, 1, \cdots, M-1 \}$ to form the indexing set $\IM$, i.e. $|\IM|=p$.
Then, the sketching matrix is abstracted as $\C = \S(:,\IM)= \S \Pii$, whereby $\Pii \in \mathbb{R}^{M\times p}$ is the equivalent sampling matrix.
As seen, $\C~\in \mathbb{R}^{M\times p}$ contains $p$ columns of $\S$ indexed by $\IM$.
In the equivalent sampling matrix $\Pii$, there are $p$ columns contains only one non-zero entry $\sqrt{M/p}$ (the remaining $(M-p)$ columns are all zeros).

By minimizing the approximation error, we further obtain another weight matrix $\W_{\text{opt}}$, i.e.
\begin{equation}\W_{\text{opt}}=\arg\min_{\W} \| \S - \C\W\C^H \|_F^2 = \C^\dag \S  {\C^\dag}^H.
\end{equation}

Since the involved pseudo-inverse and matrix multiplication incur a high complexity, we tend to the other scheme to compute $\W\simeq \W_{\text{opt}}$.
To be specific, we solve the above over-determined system by sketching both $\S$ and its random approximation.
This leads to the so-called Nystr\"om approximation \cite{tropp2017fixed,wang2019scalable}; and the weight matrix is determined by:
\begin{equation}
\W=\arg\min_{\W} \| \Pii^H(\S - \C\W\C^H)\Pii \|_F^2 =  (\Pii^H\S\Pii )^\dag.
\end{equation}

Finally, a large covariance matrix $\S$ is approximated by:
\begin{equation}
\S  \simeq \tilde{\S}
\: \triangleq \: \C \W \C^H .
\end{equation}

Relying on this low-rank representation of $\S$, its SVD can be efficiently computed via multiple SVDs of these small sketches.
Assume $\C=\U_c \pmb{\Sigma}_c \V_c^H$, the covariance matrix $\S$ would be further approximated by:
\begin{equation*}
\S  \simeq \U_c \pmb{\Sigma}_c \V_c^H  \W (\U_c \pmb{\Sigma}_c \V_c^H)^H.
\end{equation*}

After defining $\textbf{B}\triangleq \pmb{\Sigma}_c \V_c^H \W \V_c \pmb{\Sigma}_c^H$
and computing its SVD, i.e. $\textbf{B}= \U_B \pmb{\Sigma}_B \V_B^H \in \CB^{p \times p},$
we would obtain the rank-restricted approximate SVD of $\S$, i.e.
\begin{equation*}
\S \simeq \U_c \U_B \pmb{\Sigma}_B  \V_B^H  \U_c ^H = \bar{\U}_K  \pmb{\Sigma}_B \bar{\U}_K ^H,
\end{equation*}
whereby $\bar{\U}_K \triangleq \U_c \textbf{U}_B$ is one unitary matrix.
Relying on it, we finally obtain an approximated pseudo-spectrum, i.e.
\begin{equation}\label{eq:theory:music_nys}
\tilde{P}_{\textrm{{music}}}  (\theta )
\: = \: \frac{1 }{ \a (\theta )^H ( \I_M - \tilde{\U}_K \tilde{\U}_K^H ) \a (\theta ) } .
\end{equation}

The schematic flow of our fast-MUSIC method 1 is summarized in \textbf{Algorithm~\ref{alg:nystrom}}, which actually obtains the rank-restricted SVD of a large covariance and thus reduces the time complexity significantly.
Based on the above analysis, we easily see that:
Step 1 (Lines~\ref{alg:nystrom:step1:begin} to \ref{alg:nystrom:step1:end}) costs $\OM (p M )$ time.
Step 2 (Lines~\ref{alg:nystrom:step2:begin} to \ref{alg:nystrom:step2:end}) costs $\OM (p^2 M )$ time.
Combined together, the overall complexity of our fast method 1 in estimating the signal subspace is $\OM (p^2 M )$.

\begin{algorithm}[t]  
	\caption{Fast MUSIC Method -- 1.}
	\label{alg:nystrom}
	\begin{small}
		\begin{algorithmic}[1]
			\STATE {\bf Input}: a covariance matrix $\S$,
			the number of target points $K$, and over-sampling parameter $p$ ($\geq K$).
			\STATE {\bf // Step 1}: Random Sampling Method  \label{alg:nystrom:step1:begin}
			\STATE $\IM \longleftarrow$ randomly sampling $p$ indices from $\{ 0, 1, \cdots , M \}$;
			\STATE $\C \longleftarrow$ the column of $\S$ indexed by $\IM$;
			\STATE $\W \longleftarrow$ the pseudo-inverse of $\S(\IM,\IM)$;\label{alg:nystrom:step1:end}
			\STATE {\bf // Step 2}: Rank Restriction \label{alg:nystrom:step2:begin}
			\STATE Compute the SVD of $\C=\U_c \pmb{\Sigma}_c \V_c^H$;
			\STATE Compute the SVD: $\pmb{\Sigma}_c \V_c^H \W \V_c \pmb{\Sigma}_c^T = \U_B \pmb{\Sigma}_B \V_B^H$;\label{alg:nystrom:svd}
			\STATE Compute the signal subspace: $\tilde{\U}_K   \longleftarrow $ $\U_c \U_B$;\label{alg:nystrom:step2:end}
			\STATE {\bf // Step 3}: Approximate MUSIC
			\STATE Compute $\tilde{P}_{\textrm{music}}$ for all $\theta$. \label{alg:nystrom:step3}
		\end{algorithmic}
	\end{small}
\end{algorithm}


\subsection{Fast MUSIC Method -- 2}

The above random sampling method enables a highly efficient approximation of large covariance as well as the pseudo-spectrum.
In practice, this approximated pseudo-spectrum in eq. (10) is almost accurate, which suffices to attain the high-resolution DoAs estimation.
Nevertheless, one potential risk is that its error bound would be relatively weak in the worst case; see Section \ref{sec:theory:power} for the theoretical analysis.

If the precise matrix approximation was emphasized, we may resort to the other random projection technique \cite{gittens2016revisiting}.
Here, the main difference with regard to our fast method 1 is that, when abstracting a sketch $\C$, a more informative projection matrix $\Pii \in \mathbb{C}^{M \times p}$ is adopted, i.e. $\C=\S\Pii$, which will be computed in an \emph{iterative} manner.
Compared with the random sampling based matrix sketch as in \textbf{Algorithm 1}, this random projection based sketch incorporates more information of $\S$ and permits the more accurate approximation.

To begin with, we first initialize one $M \times p$ random matrix $\Pii$, whose entries are i.i.d.\ values from the Normal distribution $\NM (0, 1)$.
Then, we obtain one randomly projected sketch $\C^{(1)}=\S \Pii \in \CB^{M \times p} $.
To improve the quality of this sketching matrix, we further compute the orthonormal basis of $\C^{(1)}$,
\begin{equation*} \label{eq:power}
[\Q_c,\R_c]=\texttt{qr}(\S \Pii),~~\V^{(1)} \triangleq \Q_c,
\end{equation*}
where $\texttt{qr}(\cdot)$ denotes the QR decomposition with the time complexity $\mathcal{O}(p^2M)$; and therefore, the $M\times p$ matrix $\V $ gives another improved projection matrix.
Subsequently, we obtain one updated sketch $\C^{(2)}=\S \V^{(1)} \in \CB^{M \times p}$, and further obtain its orthonormal basis $\V^{(2)}$, i.e.
\begin{equation} \label{eq:power}
[\Q_c,\R_c]=\texttt{qr}\left( \S \V^{(1)} \right),~~\V^{(2)} \triangleq  \Q_c.
\end{equation}

The above process would be repeated for $t$ times, and finally we achieve a sufficiently good matrix sketch $\C^{(t)} = \S \V^{(t)} \in \CB^{M \times p}$.
Similar to eq. (8), the $p\times p$ weight matrix is:
\begin{equation*} \label{eq:power}
\W = ( { \V^{(t)}}^H \S \V^{(t)} )^\dag .
\end{equation*}
With $\C=\C^{(t)}$ and $\W$, we are able to obtain the projection-based matrix approximation, i.e. $\S \simeq \C \W  \C^H$.
In a similar way, the SVD of $\S$ is approximately attained by multiple SVDs on the small sketches, i.e. $\S \simeq \tilde{\U}_K \pmb{\Sigma}_B \tilde{\U}_K ^H$.
Finally, the pseudo-spectrum $\tilde{P}_{\textrm{{music}}} (\theta )$ is estimated, as in eq. (10).

The time complexity of our fast method 2 is measured by $\OM (t p M^2)$, see \textbf{Algorithm 2}.
Note that, despite the similar random projection as in the Lanczos method, here our method targets at deriving small sketches and a matrix approximation in eq. (9), which is faster than the block-Lanczos method, as shown in Table I.
In our fast method 2, the user-specific parameter $t$ trades off the complexity and the accuracy.
E.g., a large $t$ improves the projection matrix and the matrix approximation accuracy, but also increases the complexity.

One practical concern is that the noise may become non-Gaussian.
In this case, a generalized eigen-decomposition problem can be similarly formulated \cite{1995Noise}, by resorting to high-order cumulates techniques.
As a result, EVD/SVD can be still used to acquire unknown AoA.
In this regard, our new methods are readily applicable to reduce the time complexity of subspace computation in the case of non-Gaussian noises.

\section{Performance Bounds for Fast-MUSIC} \label{sec:theory}

Our primary concern is then the accuracy of the approximated pseudo-spectrum $\tilde{P}_{\textrm{{music}}} (\theta )$ relative to its exact version $P_{\textrm{{music}}} (\theta)$.
In the following, we analyze the theoretical performance of our fast-MUSIC methods, giving the relative error between the approximated pseudo-spectrum in eq. (10) and the exact pseudo-spectrum in eq. (\refeq{eq:p_music}) as standard MUSIC.

It should be noted that, from a theoretical perspective, the existing theories or error bounds developed for randomized matrix sketching are inapplicable to our concerned problem.
This is to say, the ultimate aim of our fast-MUSIC method is to approximate an exact pseudo-spectrum $P_{\textrm{music}}(\theta)$, and then to acquire the high-resolution AoAs estimation.
Unfortunately, most existing theories only bound the matrix approximation errors\cite{drineas2005nystrom,gittens2016revisiting,wang2019scalable}, e.g. $\| \S - \tilde{\S} \|_F$ or $\| \S - \tilde{\S} \|_2$.
Accordingly, such matrix norm bounds should not lend any support to our fast-MUSIC methods.

In this section, we show the relative error of our estimated pseudo-spectrum $\tilde{P}_{\textrm{music}}(\theta)$ would be also bounded, after approximating $\S$ with $\tilde{\S}$, which is more challenging than to derive the classical matrix norm bounds.
Our Theorems~\ref{thm:nys} and \ref{thm:power} establish the theoretical lower bounds for $\tilde{P}_{\textrm{{music}}}(\theta )$.
Theorem~\ref{thm:power2} establishes the theoretical upper bound for it.

\subsection{Lower/Upper Bounds of $\tilde{P}_{\textrm{music}}(\theta)$}
First, we expect to have the {\it lower bound} in the following form:
there is a bounded positive number $\alpha_l$ such that
\begin{equation}
\tilde{P}_{\textrm{{music}}} (\theta )
\: \geq \: \tfrac{1}{1 + \alpha_l} \, {P}_{\textrm{{music}}} (\theta ),
\end{equation}
for all $\theta$.
Such a lower bound is of great importance, as it guarantees the approximated pseudo-spectrum $\tilde{P}_{\textrm{{music}}} (\theta )$ shall not miss the target peaks (see the theoretical analysis in Section IV-E).
Theorems~\ref{thm:nys} and \ref{thm:power} are such lower bounds for our fast methods 1 and 2, respectively.

Second, we also expect to have an {\it upper bound} in the following form:
there is a positive number $\alpha_u$ such that
\begin{equation}
\tilde{P}_{\textrm{{music}}} (\theta )
\: \leq \: \tfrac{1}{1 - \alpha_u} \, {P}_{\textrm{{music}}} (\theta ),
\end{equation}
for all $\theta$.
Such an upper bound guarantees our approximated pseudo-spectrum $\tilde{P}_{\textrm{{music}}} (\theta )$ shall not interpret the baseline fluctuation in the non-target region as one false target peak (see the analysis Section IV-E).
Theorem~\ref{thm:power2} is such an upper bound for the fast method 2.
For the fast method 1, we do not have such an upper bound at the current stage, which would be left to the future study.

\begin{algorithm}[t]  
	\caption{Fast MUSIC Method 2.}
	\label{alg:power}
	\begin{small}
		\begin{algorithmic}[1]
			\STATE {\bf Input}: spatial covariance matrix $\S$ ($M\times M$) and target rank $K$.
			\STATE {\bf // Step 1}: {\it Random Projection}
			\STATE $\Pii \longleftarrow $ $M\times K$ matrix with entries i.i.d.\ drawn from $\NM (0, 1)$;
            \STATE Random projection based matrix sketch: $\C \longleftarrow \S \Pii$;
			\STATE Repeat $\V \longleftarrow \orth (\C )$ and $\C\longleftarrow \S \V$ for $t$ times;
            \STATE {\bf // Step 2}: {\it Randomized Matrix Approximation} \label{alg:power:step2:begin}
			\STATE $\C   \longleftarrow \S \V$
				and $\W   \longleftarrow  (\V^H \S \V )^\dag$;
		    \STATE Compute the SVD of $\C=\U_c \pmb{\Sigma}_c \V_c^H$;
			\STATE Compute the SVD: $\pmb{\Sigma}_c \V_c^H \W \V_c \pmb{\Sigma}_c^T = \U_B \pmb{\Sigma}_B \V_B^H$;\label{alg:nystrom:svd}
			\STATE Compute the signal subspace: $\tilde{\U}_K   \longleftarrow $ $\U_c \U_B$;
			\STATE {\bf // Step 3}: {\it Approximate MUSIC}
			\STATE Compute $\tilde{P}_{\textrm{music}} (\theta) $ for all $\theta$. \label{alg:power:step3}
		\end{algorithmic}
	\end{small}
\end{algorithm}

\subsection{Lower Bounds for Fast Method 1} \label{sec:theory:nys}

Let $\S_K = \U_K \Si_K \U_K^H$ be the best rank-$K$ approximation to $\S$; $\tilde{\S} = \C \W \C^H$ be randomized matrix approximation to $\S$ via uniform column sampling (i.e. $\C= \S \Pii$ and $\Pii$ is the sampling matrix);
and $\tilde{\U}$ be the orthonormal basis of $\tilde{\S}$.

\begin{theorem} \label{thm:nys}
	Let the notation be defined in the above.
	Let $\delta \in (0, 1)$ be any user-specified constant; the row coherence of $\U_K$ be defined as $\mu (\U_K)\triangleq \: \tfrac{M}{K} \max_{i} \big\| (\U_K)_{i:} \big\|_2^2 \: \in \: \big[1, \tfrac{M}{K} \big]$ \cite{candes2009exact,gittens2016revisiting,wang2016spsd}.
	For a user-specific sampling parameter
	\begin{equation*}
	p \: \geq \: 4.5 \, \mu (\U_K)\, K  \cdot \log \tfrac{K}{\delta},
	\end{equation*}
	the following relation holds with probability at least $(1-\delta)$:
	\begin{equation}
	\sqrt{\tfrac{P_{\textrm{music}} (\theta) }{\tilde{P}_{\textrm{music}} (\theta) }}
	\: \leq \: \: 1 + 2\sqrt{\tfrac{M^2}{p}}  \tfrac{ \sigma_{K+1} (\S ) }{ \sigma_{K} (\S ) } .
	\end{equation}
\end{theorem}

Thus, Theorem~\ref{thm:nys} establishes the lower bound on $\tilde{P}_{\textrm{music}} (\theta)$.
As found, the \emph{relative} approximation error is basically
\begin{equation*}
\OM \left( 2\sqrt{\tfrac{M^2}{p}}  \cdot \tfrac{ \sigma_{K+1} (\S ) }{ \sigma_{K} (\S ) } \right).
\end{equation*}
If the SNR is high, we have $\sigma_{K} (\S ) \gg \sigma_{K+1} (\S ) = \epsilon_n^2$; and therefore; the approximation error could be very small.
The proof of Theorem~\ref{thm:nys} can be found in Appendix C.

\subsection{Lower Bound for Fast Method 2} \label{sec:theory:power}

Let $\Pii$ be an $M\times p$ initial projection matrix with i.i.d.\ entries drawn from $\NM (0, 1)$.
Let the $M\times p$ matrix $\V^{(t)}$ be an improved projection matrix after $t$ iterations;
in other words, $\V^{(t)}$ is the orthonormal basis of $\S^t \Pii$.
By definition, $\C = \S \V^{(t)}$ and $\W =  ({\V^{(t)}}^H \S \V^{(t)} )^\dag$, and $\tilde{\S} = \C \W \C^H$ is an approximation to $\S$.
Let $\tilde{\U}$ be the orthonormal basis of $\tilde{\S}$.

\begin{theorem} \label{thm:power}
	Let the notation be defined in the above.
	Repeat the iteration projection for $t$ ($t \geq 0$) times.
	Let $\delta \in (0, 1)$ be any constant.
	The following relation holds with probability at least $1-\delta$:
	\begin{equation}
	\sqrt{\tfrac{P_{\textrm{music}} (\theta)}{\tilde{P}_{\textrm{music}} (\theta) }}
	\: \leq \: 1 +   \tfrac{ \sqrt{M^2 K} }{\delta} \big( \tfrac{ \sigma_{K+1} (\S ) }{ \sigma_{K} (\S ) } \big)^{t+1} .
	\end{equation}
\end{theorem}

Theorem~\ref{thm:power} establishes a lower bound for $\tilde{P}_{\textrm{music}} (\theta)$ which is obtained via random projection in Section~\eqref{eq:power}.
As seen, the relative approximation error is basically
\begin{equation*}
\OM \left(  \sqrt{M^2K}  \cdot
\big( \tfrac{ \sigma_{K+1} (\S ) }{ \sigma_{K} (\S ) } \big)^{t+1} \right) .
\end{equation*}
That means, the approximation errors exponentially converge to zero with $t$.
Even if the SNR is not very high, our fast method 2 converges to the precise estimation in several iterations (i.e. $t\geq 2$).
The proof can be found in Appendix D.

\subsection{Upper Bound for Fast Method 2} \label{sec:theory:power_up}

\begin{theorem} \label{thm:power2}
	Let the notation be defined in Section~\ref{sec:theory:power}.
	Repeat the iterative projection for $t$ ($t \geq 0$) times.
	Then, with probability at least $1-\delta$,
	\begin{align}
	\sqrt{\tfrac{P_{\textrm{music}} (\theta) }{\tilde{P}_{\textrm{music}} (\theta) }}
	\: \geq \: 1 -   \tfrac{ \sqrt{M^2 K} }{\delta} \,
	\big( \tfrac{ \sigma_{K+1} (\S ) }{ \sigma_{K} (\S ) } \big)^{t+1}.
	\end{align}
\end{theorem}

Theorem~\ref{thm:power2} establishes an upper bound for $\tilde{P}_{\textrm{music}} (\theta)$, which is also estimated via the random projection.
Similarly, the relative approximation error is basically
\begin{equation*}
\OM \left(  \sqrt{M^2 K}  \cdot
\big( \tfrac{ \sigma_{K+1} (\S ) }{ \sigma_{K} (\S ) } \big)^{t+1} \right) .
\end{equation*}
It also converges to zero exponentially.
Again, even if the SNR is not very high, few iterations suffice for the high precision.
The proof of Theorem~\ref{thm:power2} is structured into Appendix E.

\subsection{Accurate DoA Estimation}

As shown latterly, in the context of massive-MIMO radars, the target peaks of pseudo-spectrum, which are located exactly at target AoAs, are significantly higher than noise baseline,
\begin{equation*}
P_{\textrm{music}} (\theta_k) \gg P_0,
\end{equation*}
where $P_0\triangleq \max\left\{P_{\textrm{music}} (\theta')\right\}$ denotes the maximum amplitude of non-target regions, i.e. $\theta' \notin \pmb{\theta}_{K \times 1}$.
Without loss of generality, we assume it almost satisfies $P_{\textrm{music}} (\theta_k) / P_0 > \gamma $, when the SNR is not very small (e.g. $\gamma\thicksim M$; see the following Figure 3).
If the approximated pseudo-spectrum misses one target peak located at $\theta_k$, we should have
\begin{equation}
\sqrt{\tilde{P}_{\textrm{music}} (\theta_k)} \leq \sqrt{P_0}.
\end{equation}
At the same time, according to our theoretical lower bounds in eq. (13), we are supposed to have
\begin{equation*}
\tfrac{1}{1 + \alpha_l} \sqrt{ P_{\textrm{music}} (\theta_k)  } \leq \sqrt{ \tilde{P}_{\textrm{music}} (\theta_k) } \leq  \sqrt{P_0},
\end{equation*}
and therefore,
\begin{equation}
\sqrt{\gamma} \leq \sqrt{ P_{\textrm{music}} (\theta_k)/ P_0  } \leq  1 + \alpha_l .
\end{equation}

On the other hand, according to the following numerical analysis (e.g. in Figure 5), we would have $1 + \alpha_l \rightarrow  1$ and
\begin{equation}
1 + \alpha_l < \sqrt{\gamma}.
\end{equation}

By checking eq. (18) and (19), we immediately conclude this contradiction result indicates eq. (17) should be invalid.
I.e., the approximated pseudo-spectrum $\tilde{P}_{\textrm{music}} (\theta) $ should not miss the true peaks of targets.
Similarly, from the upper bound in eq. (14), we show the approximated pseudo-spectrum $\tilde{P}_{\textrm{music}} (\theta) $ should not mistake non-target AoA as one target peak, thus preventing the false alarm in automotive sensing.

\section{Numerical Simulation \& Analysis}
In the section, we examine our fast-MUSIC methods via numerical simulations; and compare them with popular subspace algorithms in high-resolution massive MIMO radars.

\subsection{Experiment Settings}

Since we focus on the subspace computation as well as subsequent AoA estimation, we directly start from a signal matrix $\Y$ of $M$ receiving elements.
In the following simulations, we configure the operation frequency to be 77GHz, as for mm-wave MIMO radar systems.
The IF bandwidth is assumed to be 80MHz.
The element spacing distance is equal to the half of wavelength.
The time complexity can be fairly measured by the CPU runtime (CPU 2.59GHz, 32GB memory), which is proportional to the required number of multiplication flops.
The averaged runtime is reported based on 100 independent realizations.

\begin{figure}[!t]\vspace{0pt}
	\begin{center}
		\subfigure[Pseudo-spectrum of the Nystr\"om method with different $p$.]{
			\label{fig:config:1}
			\includegraphics[width=64mm]{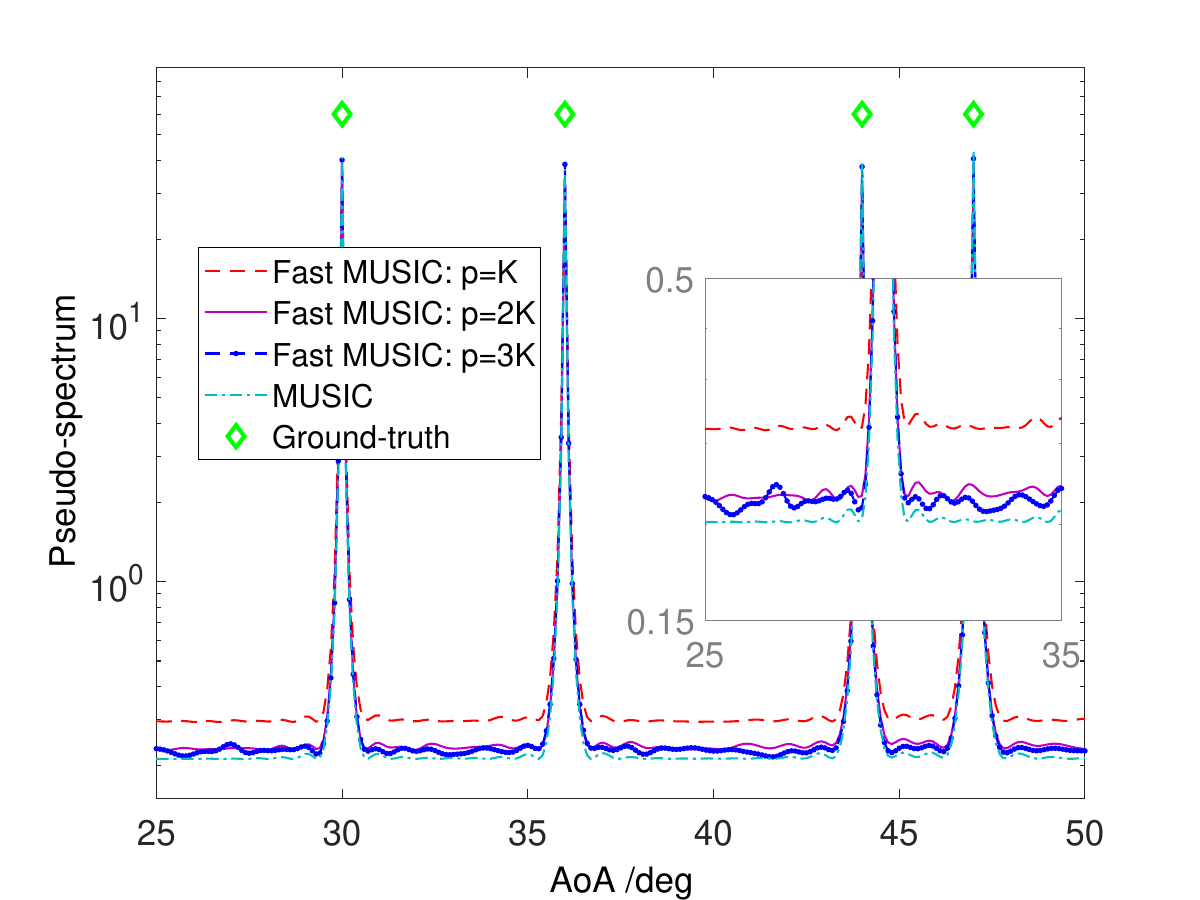}}
		\subfigure[Pseudo-spectrum of the power method with different $t$.]{
			\label{fig:config:2}
			\includegraphics[width=64mm]{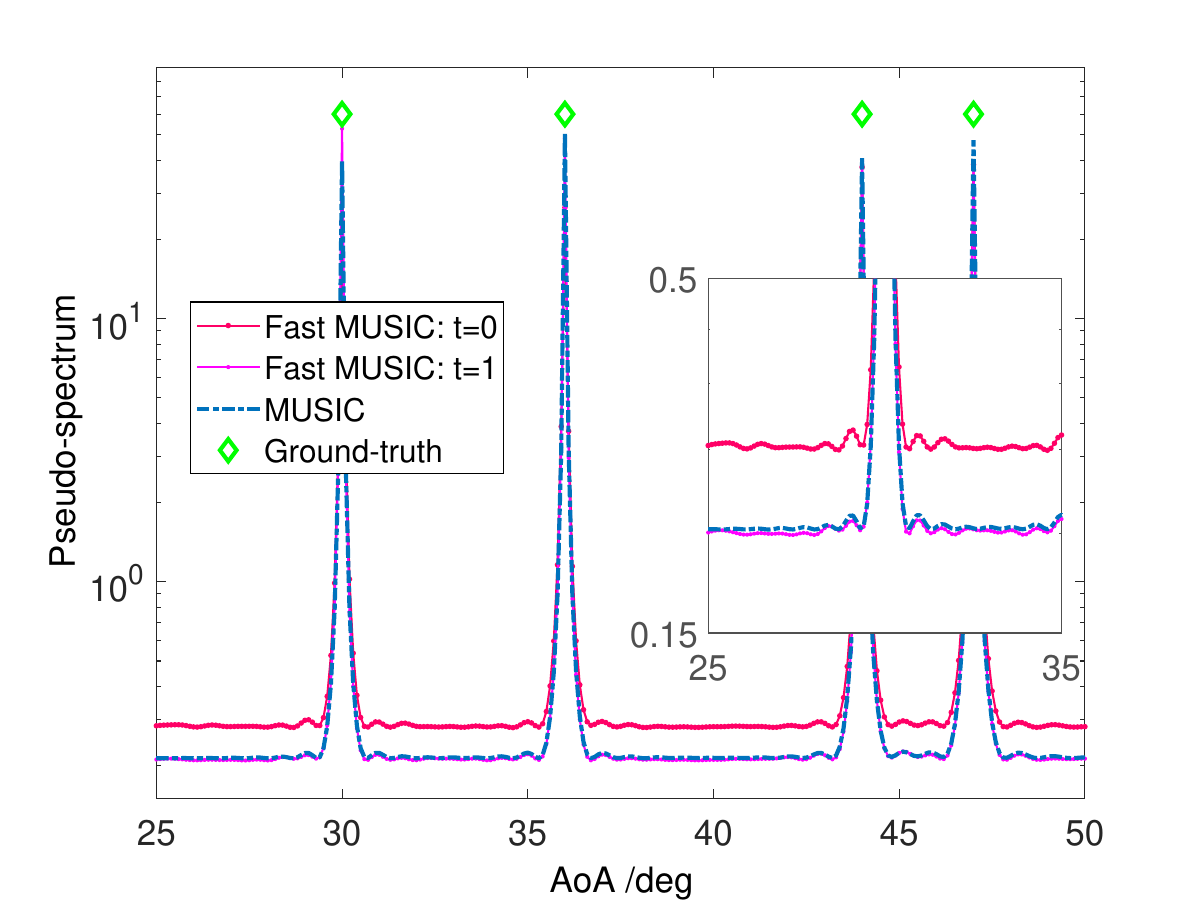}}
		\caption{We plot the pseudo-spectra of standard MUSIC and the fast-MUSIC methods with different tuning parameters.}
		\label{fig:config}
	\end{center}
\end{figure}

\subsection{Tuning Parameters}

Our fast-MUSIC method 1 (Algorithm~\ref{alg:nystrom}) involves one key user-specific parameter $p$ ($p \geq K$), which is used for trading off the accuracy and the complexity.
I.e., a large $p$ is necessary to attain the tighter error bound (as in Theorem 1), which in turns increases the time complexity.
Likewise, the fast-MUSIC method 2 (Algorithm~\ref{alg:power}) involves a tuning parameter $t$ ($\geq 1$).
In the following, we discuss their practical settings, whereby we fix $M=200$, $N=2M$, $K=4$, and SNR $= 0$ dB.

In Figure~\ref{fig:config:1}, we plot the approximated pseudo-spectrum of our fast MUSIC method 1 under different $p$.
We observed that, under three different settings ($p=K$, $2K$, and $3K$), the fast algorithm successfully obtains the accurate AoA estimation of $K$ targets (e.g. with overlapped peaks).
Define the approximation error relative to standard pseudo-spectrum as $ \sum_{l=0}^{L-1} \big[ \tilde{P}_{\textrm{music}} (\theta_l) - {P}_{\textrm{music}} (\theta_l) \big]^2$.
When $p$ is increased from $K$ to $2K$, the error drops from $0.59$ to $0.11$;
while $p$ is further increased to $3K$, the approximation error would be marginal.

In Figure~\ref{fig:config:2}, we show the approximated pseudo-spectrum of our fast method 2 under different settings of $t$; here we assume $P=K$.
When $t=2$ the approximation error is $0.03$, which is sufficiently small for the accurate AoA estimation.
In practical automotive sensing, we may directly set $t=2$.

\subsection{Complexity vs Accuracy}

We now evaluate the runtime of various popular subspace methods.
Figure~\ref{fig:runtime} shows the computation latency (seconds) against the number of antenna elements, i.e. $M$.
We assume $N=M$ and $K=10$ in the numerical analysis.
For our fast-MUSIC method 1, we assume $p= 12$ (i.e. $p=\lceil1.2 \times K \rceil $);
and in fast-MUSIC method 2, we have $t=2$ and $p= 12$. 

Despite the near-optimal accuracy in high-resolution automotive sensing, the standard MUSIC has $\OM (M^3)$ time complexity, due to the computational SVD on large covariance.
As $M$ increase, it easily becomes impractical for real-time automotive applications, by producing an intolerable latency.
When $M=1000$, the latency of SVD itself consumes more than 500 milliseconds (ms), which is formidable to real-time sensing of unmanned aircrafts/vehicles;
similar problems persist in matrix inverse and Lanczos methods.
Here, we do not compare the ESPRIT method, as it relies on SVD and has exactly the same runtime as MUSIC in subspace computation.

\begin{figure}[!t]
	\centering
	\includegraphics[width=6.3cm]{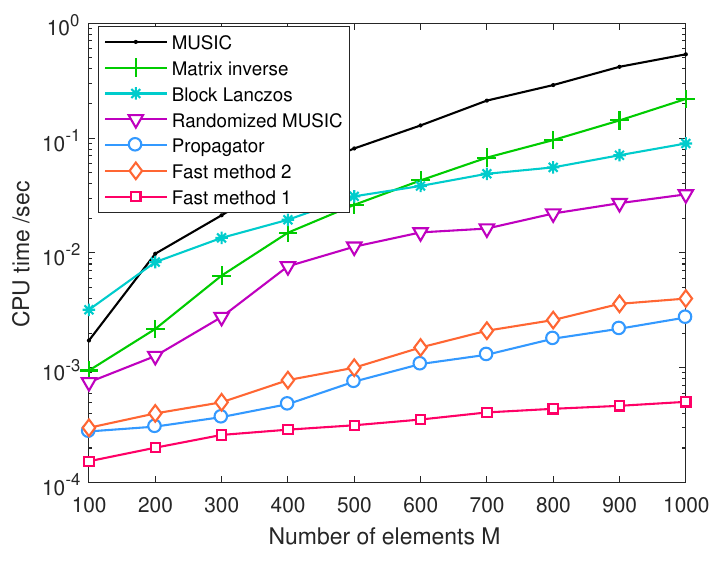}
	\caption{Wall-clock runtime (second) against the number of antennas, $M$.
		We calculate only the runtime for computing the signal subspace which is the major bottleneck in automotive radar.
		Here we set $N=M$ and SNR $=0$~dB.}
	\label{fig:runtime}
\vspace{-3mm}
\end{figure}


As shown in Figure~\ref{fig:runtime}, our fast-MUSIC method 1 is the most efficient, with its time complexity linearly scalable with $M$, i.e. $\OM (p^2 M )$.
For a large array $M=1000$, its required computation latency is only 0.5 ms; while MUSIC consumes 500 ms. So, it would be $1000$ times faster than the standard MUSIC.
It is noteworthy that, although our fast method 2 theoretically has a time complexity $\OM (tp M^2)$, its time cost, as discussed, is dramatically lower than the block-Lanczos algorithm and the other Randomized MUSIC \cite{Li2020Fast}.
As discussed, the iteration number in fast method 2 is set to $2$, while $p\sim \mathcal{O}(K)$.
As validated by Figure~\ref{fig:runtime}, its time complexity is much lower than the Lanczos method, whereby the parameter $p_M$ grows as $\OM(\log M)$ (e.g. $p_M=7$ for $M=1024$).
Its processing latency is only 4 ms when $M=1000$, only slightly higher than the Propagator method (2.7 ms).
As such, our fast method 2 is also applicable to real-time automotive sensing.

\subsection{Error Bound vs Actual Error}\label{sec:exp:bounds}


\begin{figure}[!t]\vspace{0pt}
	\begin{center}
		\subfigure[ ] {
			\label{fig:bound:1}
			\includegraphics[width=64mm]{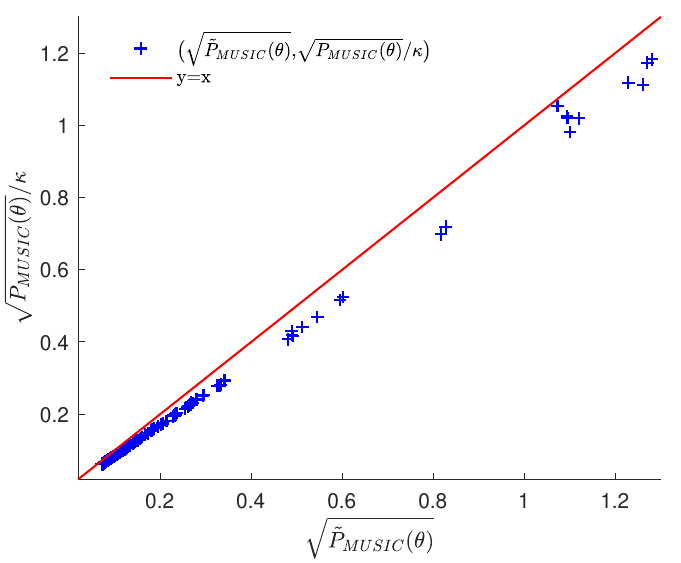}}
		\subfigure[ ] {
			\label{fig:bound:2}
			\includegraphics[width=64mm]{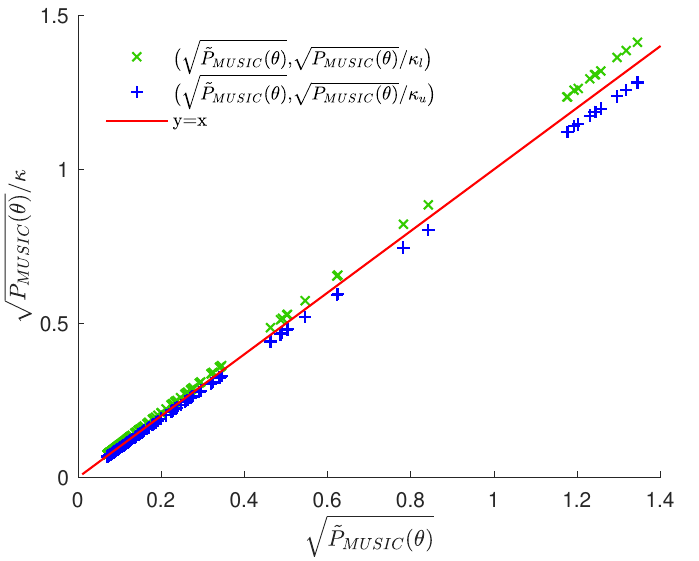}}
		\caption{
Comparing the theoretical bounds (Theorems~\ref{thm:nys}, \ref{thm:power}, and \ref{thm:power2}) with the numerical results. (a) Theoretical bound of fast method 1. (b) Theoretical bounds of fast method 2, \emph{green cross} -- upper bound, \emph{blue cross} -- lower bound. The constants $\kappa$, $\kappa_l$, and $\kappa_u$ are defined in Section~\ref{sec:exp:bounds}.}

		\label{fig:bound}
	\end{center}
\vspace{-4mm}
\end{figure}

We then study the accuracy of our theoretical error bounds (Theorems~\ref{thm:nys}, \ref{thm:power}, and \ref{thm:power2}), by comparing the pseudo-spectrum ${\tilde{P}}_{\text{music}}(\theta)$ approximated by our fast-MUSIC with the exact one ${P}_{\text{music}}(\theta)$.
We will show how much the theoretical bounds under-estimate or over-estimate an exact ${P}_{\text{music}}(\theta)$.
In this simulation, we have $M=N=400$, $p=K=12$ and $t=3$.

\subsubsection{\textbf{Fast MUSIC -- Method 1}}
We first study the lower bound for our fast-MUSIC method 1, which is premised on random sampling.
We denote the righthand side of the bound in Theorem~\ref{thm:nys} by $\kappa = 1 + \alpha_l$.
Thus, we have
\begin{equation} \label{eq:nys_lower}
\sqrt{\tilde{P}_{\textrm{music}}(\theta)} \geqslant \sqrt{P_{\textrm{music}}(\theta)} /\kappa.
\end{equation}
We repeat our numerical simulations to calculate a set of $\tilde{P}_{\textrm{music}}(\theta)$ and ${P}_{\textrm{music}}(\theta)$.
In the simulation, we assume $K=9$, $M=200$, $N=2M$ and SNR=1dB.
We set $p=10K\text{log}(K/\delta)$ ($\delta=0.01$) when deriving the theoretical lower bound.
Interestingly, our fast method 1 attains the accurate pseudo-spectrum, even if the user-specific parameter is $p \thicksim \mathcal{O}(K)$.
In this analysis, $K$ unknown DoAs are randomly generated in the angle range [0,~80] deg.

In Figure~\ref{fig:bound:1}, we plot $\sqrt{P_{\textrm{music}}(\theta)} /\kappa$ (the $y$-axis) against $\sqrt{\tilde{P}_{\textrm{music}}(\theta)}$ (the $x$-axis) as the blue crosses.
Ideally, if the lower bound is tight ($\alpha_l \rightarrow 0$, $\kappa_l \rightarrow 1$), then the blue crosses should fall on the line $y=x$.
Since the lower bound (righthand side of eq. \eqref{eq:nys_lower}) would over-estimate $\sqrt{\tilde{P}_{\textrm{music}}(\theta)}$, the blue crosses can fall below the line $y=x$ with the probability approaching 1.
That means, the empirical results match the derived theoretical result, i.e. no blue cross is above $y=x$.

If $\theta'$ falls in the non-target region ($\theta'\notin \pmb{\theta}_{K \times 1}$), the normalized pseudo-spectrum is very small, e.g. ${P}_{\textrm{music}}(\theta')\thicksim \mathcal{O}(1/M)$.
Figure~\ref{fig:bound:1} shows that our lower bound is very tight in the non-target region.
As seen, the blue crosses are almost on the line $y=x$ and hence $\alpha_l\rightarrow 0$, e.g. ${P}_{\textrm{music}}(\theta')<0.2$.
Following this, the approximated pseudo-spectrum $\tilde{P}_{\textrm{music}}(\theta')$ should not interpret non-target peaks as targets.

If $\theta_k$ is in the target region ($\theta_k \in \pmb{\theta}_{K \times 1}$), the normalized pseudo-spectrum is relatively large, ${P}_{\textrm{music}}(\theta_k)\thicksim 1$.
Figure~\ref{fig:bound:1} shows that in the target region, the empirical results (blue cross) deviate slightly from the lower bound, indicating our lower bound is less tight in this region (e.g. ${P}_{\textrm{music}}(\theta_k) \thicksim 1$).
However, this would be fine for realistic applications.
I.e. $\tilde{P}_{\textrm{music}}(\theta_k)$ for target AoAs is sufficiently large to be detected via its peak amplitude, i.e. $1+\alpha_l \ll \sqrt{\gamma}$.

\begin{figure}[!t]
	\centering
	\includegraphics[width=7cm]{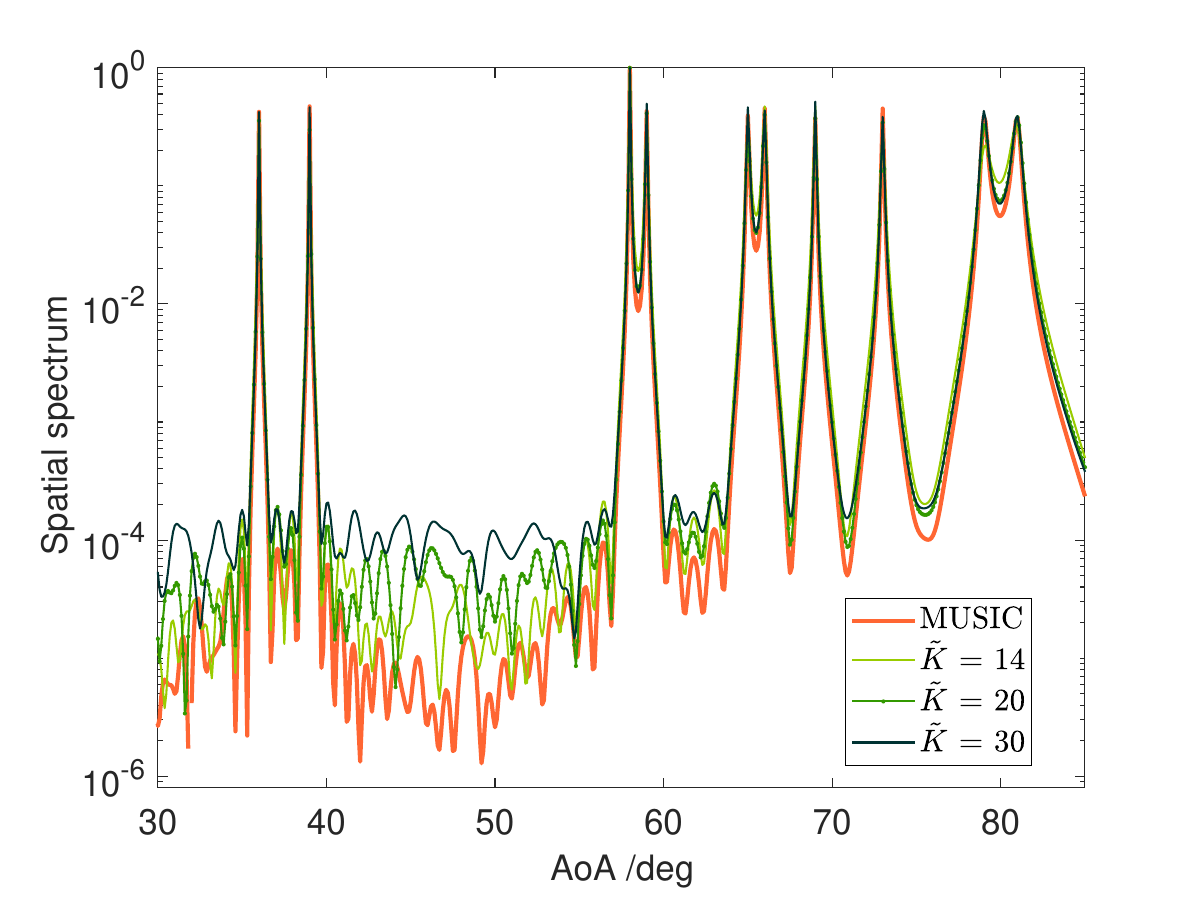}
	\caption{Robustness of our fast MUSIC algorithm to inexact guess of $K$.
		Here, the actual number of targets is $K=10$.
		The red curve plots MUSIC with the exact $K$.
		The other four curves plot our fast MUSIC algorithm with different guesses of $K$.}
	\label{fig:K}
\end{figure}

\subsubsection{\textbf{Fast MUSIC -- Method 2}}
We study the theoretical error bounds of our fast method 2, as in Theorems~\ref{thm:power} and \ref{thm:power2}.
We denote the righthand side of the lower bound in Theorem~\ref{thm:power} by $\kappa_l = 1 + \alpha_l$;
and the righthand side of the upper bound in Theorem~\ref{thm:power2} by $\kappa_u = 1 + \alpha_u$.
Then, we expect
\begin{align*}
\sqrt{\tilde{P}_{\textrm{music}}(\theta)} \geq \sqrt{P_{\textrm{music}}(\theta)} /\kappa_l ,~~
\sqrt{\tilde{P}_{\textrm{music}}(\theta)} \leq \sqrt{P_{\textrm{music}}(\theta)} /\kappa_u .
\end{align*}

In the numerical analysis, we assume $K=9$, $M=200$, $N=2M$, $\delta=0.1$, SNR=1dB, and $t=2$.
In Figure~\ref{fig:bound:2}, we show $\sqrt{P_{\textrm{music}}(\theta)} /\kappa_l$ against $\sqrt{\tilde{P}_{\textrm{music}}(\theta)}$ as blue crosses,
and $\sqrt{P_{\textrm{music}}(\theta)} /\kappa_u$ against $\sqrt{\tilde{P}_{\textrm{music}}(\theta)}$ as green crosses.
As shown, the blue crosses and green crosses lie almost on the line $y=x$, indicating that our theoretical lower and upper bounds are very tight, i.e. $\alpha_l\rightarrow 0$ and $\alpha_u \rightarrow 0$.
Combined with the theoretical results in Section IV-E, we conclude there will be no false/missed peaks in the approximated pseudo-spectrum.

\subsection{Robustness to Inaccurate $K$}
Similar to standard MUSIC or ESPRIT, our fast-MUSIC requires the prior knowledge of the number of targets, $K$.
Unfortunately, in practice $K$ would become unknown.
Thus, we have to use an over-estimate of $K$ as the input.

In Figure~\ref{fig:K}, we study the effects on the approximated pseudo-spectrum from an inexact guess of $K$.
In this simulation, we assume $M=200$, $N=400$, $K=10$, $\textrm{SNR}=0$ dB and $p=\lceil1.2\times \tilde{K}\rceil$.
Here, the standard MUSIC is assumed with the exact knowledge, i.e. $K=10$.
As in Figure~\ref{fig:K}, our fast MUSIC method 1 is practically robust to an over-estimate of $K$.
In other words, once the guess value $\tilde{K}$ is larger than $K$, our fast-MUSIC produces almost the same pseudo-spectrum (e.g. with the overlapped peaks).

Another more promising solution to combat this prior uncertainty of target numbers is to apply the \emph{pre-estimation} scheme \cite{2020MIMO}, by directly acquiring \emph{a priori} estimation of $K$.
For example, both the Akaike information criterion (AIC) \cite{Akaike1974A} or minimum description length (MDL) \cite{Barron1998The} could be used to estimate the number of targets $K$.
On this basis, we may configure the sampling length $p$ more reliably, which introduces a slight increase on time complexity but excludes the potential risk of underestimating $K$.

\begin{figure}[!t]\vspace{0pt}
	\begin{center}
		\subfigure[$M=200$, $N=800$, and SNR $= 0$ dB.]{
			\label{fig:normalized:1}
			\includegraphics[width=65mm]{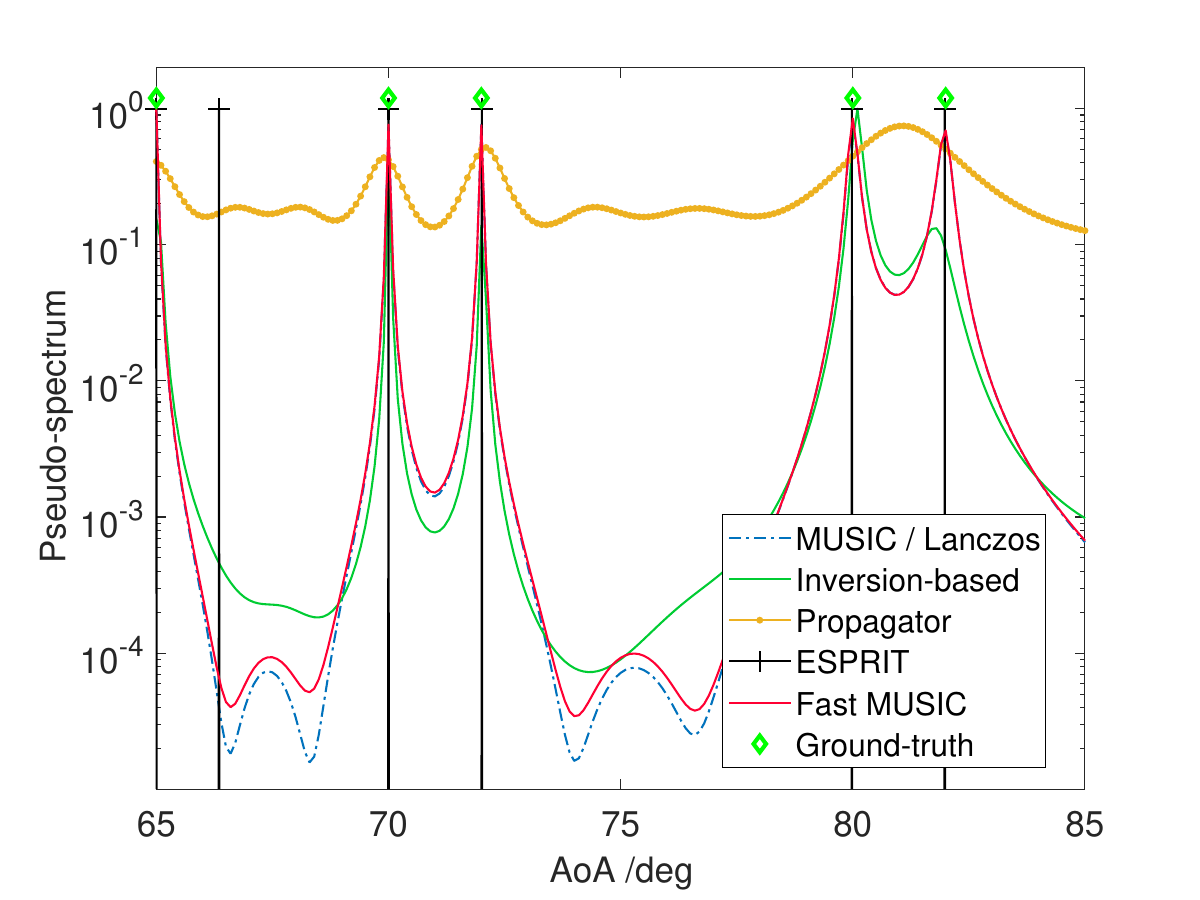}}
		\subfigure[$M=N=200$ and SNR $= 0$ dB.]{
			\label{fig:normalized:2}
			\includegraphics[width=65mm]{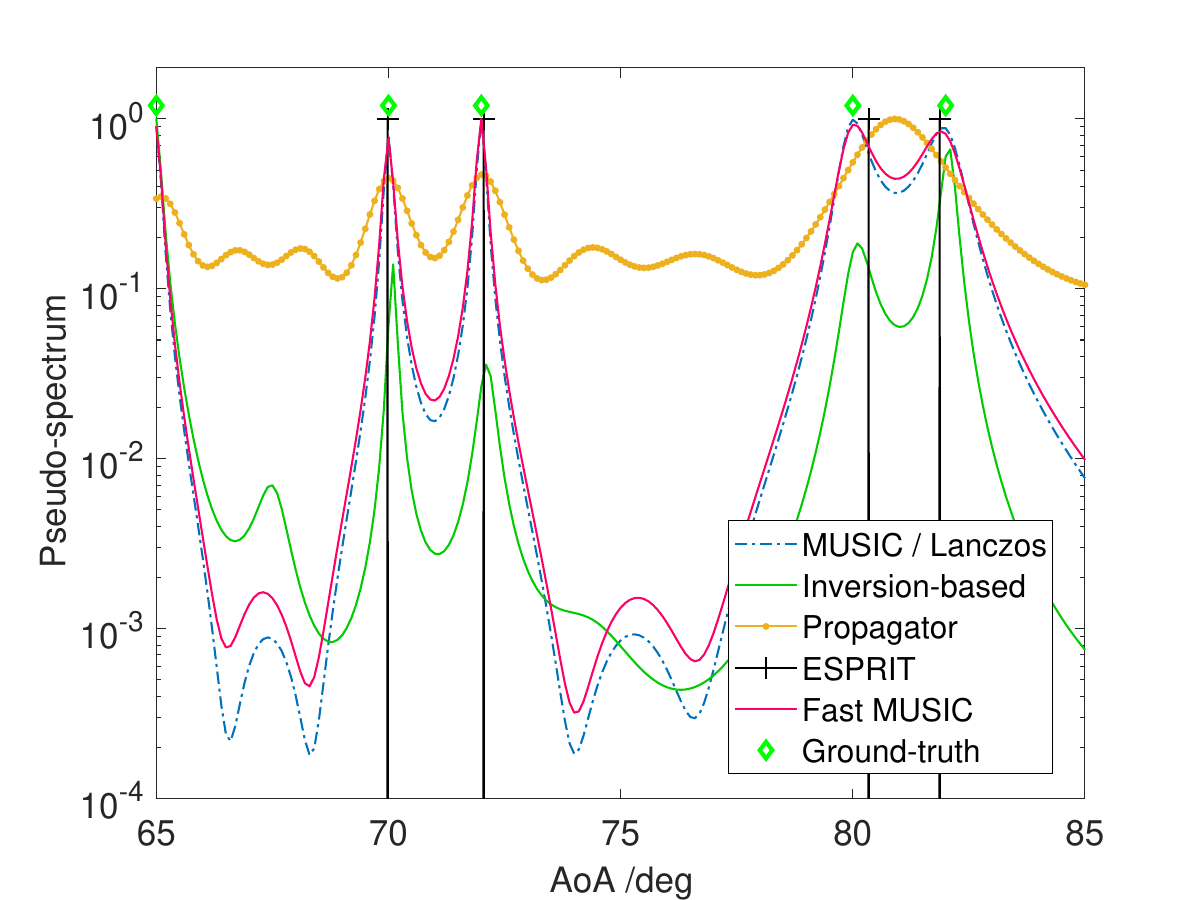}}
		\caption{Estimated pseudo-spectrum $\check{P}(\theta)$ of various algorithms.
			Here, the proposed algorithm refers to the fast-MUSIC method with $p=12$.}
		\label{fig:normalized}
	\end{center}
\end{figure}

\subsection{Accuracy in AoA Estimation}

We now evaluate the AoA estimation accuracy and compare fast-MUSIC (e.g. Fast Method 1) with the existing subspace algorithms.
In Figure~\ref{fig:normalized}, the normalized pseudo-spectrum is
\begin{equation*}
\check{P}(\theta)
\: = \: \tfrac{P(\theta)- \min_{\theta} \{ P(\theta) \} }{\max_{\theta} \{ P(\theta) \}  - \min_{\theta} \{ P(\theta) \} }
\end{equation*}
where $P (\theta)$ can be the spectrum of MUSIC, ESPRIT, Propagator, etc.
In the numerical simulation, we assume $M=200$ and $K=9$ and then evaluate the relative accuracy under different $N$.
Here, we focus on our fast method 1 with $p=12$.

As illustrated by Figure~\ref{fig:normalized} only showing AoAs that fall into [65 86] deg), among all counterpart methods, MUSIC and block-Lanczos could attain the near-optimal pseudo-spectrum, by exploiting the exact subspace information, which thus acquire the highly accurate AoAs estimation.
Our fast-MUSIC obtains an approximated pseudo-spectrum that closely tracks the standard MUSIC.
Moreover, the pseudo-spectrum of both MUSIC and our fast-MUSIC will not change, when the snapshot number $N$ drops from $N=800$ to $N=200$.
In comparison, another ESPRIT method attains less accurate AoAs.
The estimation error (i.e. false targets nearby 65 degrees) seems to be inevitable, especially when the exact number of sources (i.e. $K$) remains unknown.

Despite its low complexity, the Propagator method produces the less accurate pseudo-spectrum.
For one thing, the fluctuated baseline is relatively high, which may produce false alarms in low SNR cases.
For another, it gets the less accurate estimation when the target DoAs surpass 70 deg.
For example, it may misinterpret two targets (between 80 and 85 degrees) as one fake object, resulting in the significant estimation error.
Note that, a modified Propagator method was proposed to improve the accuracy \cite{Tayem2005L}, which, however, requires one special array structure (e.g. $L$-shape array) and tends to be less attractive for compact automotive MIMO radars.


For another matrix-inversion method, even through the time complexity in approximating noise subspace is reduced to some extent, it still has drawbacks in realistic applications.
Its estimation accuracy may be guaranteed only in high SNRs, subject to the sufficiently accurate approximation of subspace.
In the case of $N \leq M$ or strong noise ($\sigma_K(\S) \simeq \sigma_n^2$), it may be unstable and produces false objects or inaccurate AoAs.

\begin{figure}[!t]
	\centering
	\includegraphics[width=6.4cm]{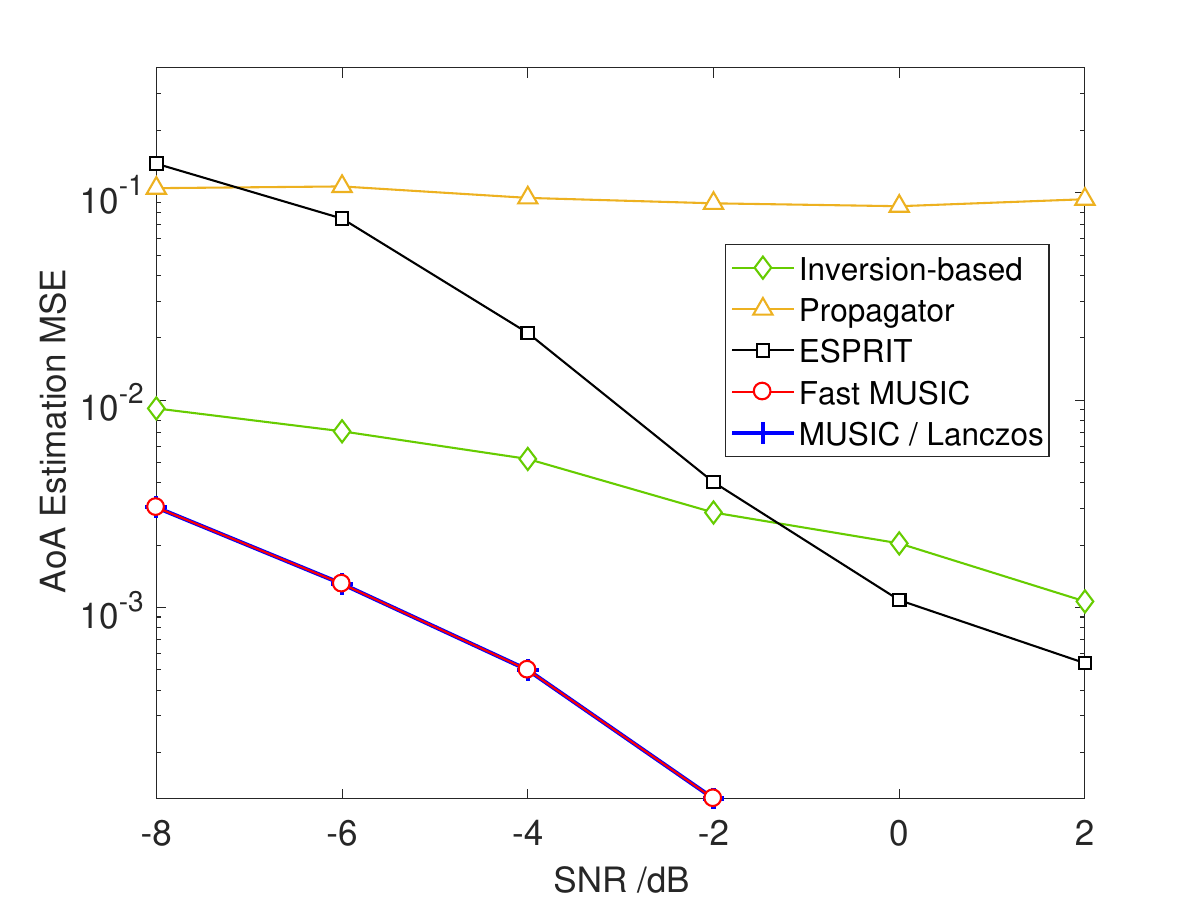}
	\caption{Plot of AoA estimation MSE against SNRs. }
	\label{fig:snr}
\end{figure}

Except for the linear complexity and real-time implementation, our fast-MUSIC is highly accurate and stable.
Unlike its counterparts that have to scarify accuracy for complexity, our new method resolves the long-standing dilemma between high resolution and low complexity in MIMO radar signal processing.
It thus achieves the high-resolution AoA estimation at a scalable complexity, which provides also the great potential to massive MIMO signal processing.

One practical challenge for our fast-MUSIC is that, when multiple targets of different SNRs are taken into accounts, the detection accuracy may be affected slightly.
From Figure 3, we observe the targets with low SNRs may be probably missed, when the number of elements is relatively small.
Fortunately, owing to the substantial processing gain of massive MIMO array and the subspace methods, this baseline fluctuation is extremely low, i.e. $\mathcal{O}(1/M)$ of the maximum peak.
In other words, those low SNR targets can be still detected and their AoAs can be estimated, provided the smallest SNR was larger than $\mathcal{O}(1/M)$ of the maximum SNR.

\subsection{MSE in AoA Estimation}
Finnaly, we evaluate the AoA estimation error of various methods in the context of massive-MIMO radar.
The mean square error (MSE) of estimated AoAs is measured by:
\begin{align}
\text{MSE}  & \triangleq \mathbb{E}_{\theta_k \in \mathcal{U}(0,\pi/2)} \left\{\frac{1}{K}\sum\nolimits_{k=1}^{K} \| \theta_k - \hat{\theta}_k \|_2^2\right\}.
\end{align}
Here, $\theta_k$ denotes the ground truth AoA of the $k$th target point, while $\hat{\theta}_k$ denoted its estimation.
In Figure~\ref{fig:snr}, we show the MSE performance of various subspace methods.
We set $M=200$, $K=10$ and $N=220$; our fast method 1 is used with $p=11$.
The AoAs estimation MSE of all subspace methods would be reduced, as SNR increases.
As noted, SNR affects the spectral gap $\tfrac{ \sigma_{K+1} (\S) }{ \sigma_K (\S) }$, and hence the approximation accuracy of the MUSIC pseudo-spectrum.
However, it is noteworthy that the MSE of estimated AoA has a \emph{non-linear} relation with the approximated pseudo-spectrum.
I.e. only when the error of estimated pseudo-spectrum surpasses one threshold, can the attained AoA suffer a serious deviation \cite{1995The,vallet2015performance}.
In this regard, our fast-MUSIC attains the same MSE as standard MUSIC \footnote{Note that, when the number of elements is relatively small (e.g. $M\leq$100), the MSE of our fast-MUSIC will be still comparable to standard MUSIC.}, even with the pseudo-spectrum approximation error.
That is to say, it scarifies no accuracy in estimated AoA, while substantially reducing the time complexity.

\section{Conclusion}

We investigate the high-resolution target detection and AoA estimation problem in the automotive massive-MIMO radar.
Existing subspace methods (e.g. MUSIC and ESPRIT) are computationally expensive, due to the complex decomposition of a large covariance matrix.
Other simplified schemes, e.g. matrix-inverse and propagator methods, lead to the degraded accuracy.
As one long-standing dilemma in MIMO radar signal processing, the high-resolution estimation and the real-time computation remain mutually exclusive.

To enable real-time and high-resolution environmental sensing, we leverage on randomized matrix approximation and develop two fast-MUSIC methods, whereby a large covariance was approximated by three small sketches that are abstracted via random sampling/projection.
We theoretically show that, in high SNRs, our fast-MUSIC is almost as good as the near-optimal MUSIC.
The numerical study also demonstrates that, whilst our method is nearly as accurate as MUSIC in acquiring unknown AoA, it is dramatically faster than standard MUSIC by orders-of-magnitude.
As such, our fast-MUSIC enjoys both the linear complexity and the high accuracy, breaking the theoretical bottleneck in MIMO radar signal processing.
It would provide the great potential to high-resolution and real-time massive-MIMO radars in the emerging automotive applications.
In the future study, we may consider the extension of this new concept of  randomized matrix sketching to the other methods, e.g. subspace methods (e.g. Capton method) and the optimal ML method.

%
\section*{Acknowledgment}
This work was supported by Natural Science Foundation of China (NSFC) under Grants U1805262.

\section*{Appendix}

Here, we show the detailed proofs of Theorems~\ref{thm:nys}, \ref{thm:power}, and \ref{thm:power2}.
For simplicity, we leave out the $\theta$ in ${P}_{\textrm{music}} (\theta)$,  $\tilde{P}_{\textrm{music}} (\theta)$, and  $\a (\theta)$.

\subsection{Analysis of Uniform Sampling}

An $M\times p$ ($p < M$) matrix $\Pii$ is called a uniform sampling matrix if its columns are sampled from the columns of $\frac{\sqrt{ M}}{\sqrt{p}} \I_M$ uniformly at random.
For any $N\times M$ matrix $\X$, the multiplication $\X \Pii$ ($N\times p$) contains $p$ uniformly sampled columns of $\X$.

\begin{lemma}  \label{lem:uniform1}
	Let $\Pii$ be an $M\times p$ uniform sampling matrix.
	Then
	\begin{equation*}
	\big\| \Pii \big\|_2^2 \: = \: M/p .
	\end{equation*}
\end{lemma}

\begin{lemma} [\!\cite{wang2016spsd,woodruff2014sketching}]  \label{lem:uniform2}
	Let $\U_K$ be an $M\times K$ matrix with orthonormal columns.
	Let the row coherence of $\U_K$ be defined as $\mu (\U_K)= \: \tfrac{M}{K} \max_{i} \big\| (\U_K)_{i:} \big\|_2^2
\: \in \: \big[1, \tfrac{M}{K} \big]$.
	Let $\Pii$ be an $M\times p$ uniform sampling matrix.
	For
	\begin{equation*}
	p \: \geq \: \frac{(6+2\eta)\mu (\U_K) K  }{3\eta^2} \log \frac{K}{\delta}
	\end{equation*}
	the $K$ singular values of $\U_K^H \Pii \Pii^H \U_K$ are at least $1-\eta$ with probability at least $1-\delta$.
\end{lemma}

\subsection{Analysis of Gaussian Projection}

\begin{lemma} [\cite{vershynin2010introduction}]\label{lem:gauss1}
Let $\Pii $ be an $M \times K$ matrix whose entries are i.i.d.\ drawn from $\NM (0, 1)$.
If $M$ is substantially larger than $K$, the spectral norm of $\Pii$ is bounded by
\begin{equation*}
\big\| \Pii \big\|_2
\: \leq \: \sqrt{M } + \sqrt{K } + \OM (1)  , \qquad \textrm{almost surely}.
\end{equation*}
\end{lemma}

\begin{lemma} [\cite{rudelson2008littlewood,tao2010random}] \label{lem:gauss2}
	Let $\G$ be a $K \times K$ matrix whose entries are i.i.d.\ drawn from $\NM (0, 1)$.
	Then the smallest singualr value of $\G$ satisfies
	\begin{equation*}
	\sigma_{K} ( \G )
	\: \geq \: \tfrac{\delta }{\sqrt{ K} }
	\end{equation*}
	with probability at least $1-\delta - o(1) $.
\end{lemma}

\subsection{Proof of Theorem \ref{thm:nys}}

\begin{lemma} \label{lem:nys1}
	Let the notation be defined in Section~\ref{sec:theory:nys}.
	Let $\Pii$ be any matrix with $\rk (\U_K^H \Pii) \geq K$.
	Then, for all $\a$,
	\begin{align*}
	& \big\| \U_K \U_K^H \a - \tilde{\U} \tilde{\U}^H  \U_K \U_K^H \a  \big\|_2   \\
	& \leq \: \Big\| (\S - \S_K) \Pii \big( \U_K^H \Pii \big)^\dag \Si_K^{-1} \U_K^H \a   \Big\|_2 ,
	\end{align*}
\end{lemma}

\begin{proof}
	Let $\tilde{\U}$ be the orthonormal bases of $\tilde{\S}$.
	Lemma 12 of \cite{wang2016spsd} shows that
	\begin{equation*}
	\big( \C \W^\dag \C^H \big) \big( \C \W^\dag \C^H \big)^\dag \C
	\: = \: \C .
	\end{equation*}
	Thus $\C = \S \Pii$ is in the subspace spanned by the columns of $\tilde{\U}$, and therefore,
	\begin{equation} \label{eq:lem:nys1:1}
	\tilde{\U} \tilde{\U}^H \S \Pii
	\: = \: \S \Pii .
	\end{equation}
	It is not hard to prove that
	\begin{equation*}
	\tilde{\U}^H  \U_K \U_K^H \a
	\: = \: \argmin_{\x} \big\| \U_K \U_K^H \a - \tilde{\U} \x \big\|_2.
	\end{equation*}
	Thus, the inequality holds for all $\tilde{\x}$:
	\begin{align}\label{eq:lem:nys1:2}
	& \big\| \U_K \U_K^H \a - \tilde{\U} \tilde{\U}^H  \U_K \U_K^H \a  \big\|_2 \nonumber \\
	& = \: \min_{\x} \big\| \U_K \U_K^H \a - \tilde{\U} \x \big\|_2
	\: \leq \: \big\| \U_K \U_K^H \a - \tilde{\U} \tilde{\x} \big\|_2 .
	\end{align}
	We artifically construct
	\begin{equation*}
	\tilde{\x}
	\: = \: \tilde{\U}^H \S \Pii \big( \U_K^H \Pii \big)^\dag \Si_K^{-1} \U_K^H \a .
	\end{equation*}
	It follows from \eqref{eq:lem:nys1:1} that
	\begin{align} \label{eq:lem:nys1:3}
	& \tilde{\U} \tilde{\x}
	\: = \: \tilde{\U} \tilde{\U}^H \S \Pii \big( \U_K^H \Pii \big)^\dag \Si_K^{-1} \U_K^H \a \nonumber  \\
	& = \: \S \Pii \big( \U_K^H \Pii \big)^\dag \Si_K^{-1} \U_K^H \a  \nonumber \\
	& = \: \big( \S_K + \S - \S_K \big) \Pii \big( \U_K^H \Pii \big)^\dag \Si_K^{-1} \U_K^H \a .
	\end{align}
	Since $\rk (\U_K^H \Pii) \geq K$, we have $(\U_K^H \Pii) ( \U_K^H \Pii )^\dag = \I_K$.
	Thus
	\begin{align}  \label{eq:lem:nys1:4}
	&\S_K \Pii \big( \U_K^H \Pii \big)^\dag \Si_K^{-1} \U_K \a \nonumber \\
	& = \: \U_K \Si_K \U_K^H \Pii \big( \U_K^H \Pii \big)^\dag \Si_K^{-1} \U_K^H \a \nonumber \\
	& = \: \U_K \Si_K \Si_K^{-1} \U_K^H \a
	\: = \: \U_K  \U_K^H \a .
	\end{align}
	It follows from \eqref{eq:lem:nys1:3} and \eqref{eq:lem:nys1:4} that
	\begin{align} \label{eq:lem:nys1:5}
	& \tilde{\U} \tilde{\x}
	\: = \: \big[ \S_K + (\S - \S_K) \big] \Pii \big( \U_K^H \Pii \big)^\dag \Si_K^{-1} \U_K^H \a \nonumber \\
	& = \: \U_K  \U_K^H \a  + (\S - \S_K) \Pii \big( \U_K^H \Pii \big)^\dag \Si_K^{-1} \U_K^H \a  .
	\end{align}
	It follows from \eqref{eq:lem:nys1:2} and \eqref{eq:lem:nys1:5} that
	\begin{align*}
	& \big\| \U_K \U_K^H \a - \tilde{\U} \tilde{\U}^H  \U_K \U_K^H \a  \big\|_2
	\: \leq \: \big\| \U_K \U_K^H \a - \tilde{\U} \tilde{\x} \big\|_2 \\
	& = \: \Big\| (\S - \S_K) \Pii \big( \U_K^H \Pii \big)^\dag \Si_K^{-1} \U_K^H \a   \Big\|_2 ,
	\end{align*}
	by which the lemma follows.
\end{proof}

\noindent
{\bf Complete the proof of Theorem~\ref{thm:nys}:}
\begin{proof}
  	The $M\times p$ matrix $\Pii$ is a uniform sampling matrix.
	Then $\C = \S \Pii$ and $\W = \Pii^T \S \Pii$.
	Lemma~\ref{lem:uniform1} shows the spectral norm of $\Pii$ is bounded by
	\begin{equation*}
	\big\| \Pii \big\|_2
	\: \leq \: \sqrt{M/p} .
	\end{equation*}
	Because $\U_K$ has orthonormal columns, Lemma~\ref{lem:uniform2} shows that for $
	p \: \geq \: \frac{(6+2\eta)\mu (\U_K) K  }{3\eta^2} \log \frac{K}{\delta},$
	the $K$-th singular value of $\U_K^H \Pii$ is bounded by
	\begin{equation*}
	\sigma_{K} ( \U_K^H \Pii )
	\: \geq \: \sqrt{1 - \eta }
	\end{equation*}
	with probability at least $1-\delta$.
	Let follows from Lemma~\ref{lem:nys1} that
	\begin{align*}
	& \big\| \big( \I_M - \tilde{\U} \tilde{\U}^H  \big) \U_K \U_K^H \a  \big\|_2  \\
	& \leq \: \big\| \U_K \U_K^H \a - \tilde{\U} \tilde{\U}^H  \U_K \U_K^H \a  \big\|_2   \\
	& \leq \: \Big\| (\S - \S_K) \Pii \big( \U_K^H \Pii \big)^\dag \Si_K^{-1} \U_K^H \a   \Big\|_2 \\
	& \leq \: \big\| (\S - \S_K) \big\|_2 \, \big\| \Pii \big\|_2 \, \big\| \big( \U_K^H \Pii \big)^\dag \big\|_2 \, \big\| \Si_K^{-1} \big\|_2\, \big\| \U_K^H \a   \big\|_2 \\
	& \leq \: \frac{ \sigma_{K+1} (\S ) }{ \sigma_{K} (\S ) } \,
	\frac{\sqrt{{M}/{p} } }{ \sqrt{1- \eta} } \, \big\| \U_K^H \a   \big\|_2  .
	\end{align*}
	We set $\eta = 0.75$.
	Then, for
	\begin{equation*}
	p \: \geq \: 4.5 \mu (\U_K) \, K \cdot  \log \tfrac{K}{\delta},
	\end{equation*}
	it holds with probability at least $1-\delta$ that
	\begin{align*}
	& \big\| \big( \I_M - \tilde{\U} \tilde{\U}^H  \big) \U_K \U_K^H \a  \big\|_2
	\: \leq \: 2 \sqrt{\tfrac{ M}{p}}  \, \tfrac{ \sigma_{K+1} (\S ) }{ \sigma_{K} (\S ) } \, \big\| \U_K^H \a \big\|_2  .
	\end{align*}
	It follows that
	\begin{align*}
	& \big\| \big(\I_M - \tilde{\U} \tilde{\U}^H \big) \a \big\|_2 \\
	& = \: \big\| \big(\I_M - \tilde{\U} \tilde{\U}^H \big) \big( \U_K \U_K^H + \U_{-K} \U_{-K}^H \big) \a \big\|_2 \\
	& \leq \: \big\| \big(\I_M - \tilde{\U} \tilde{\U}^H \big) \U_K \U_K^H  \a \big\|_2
	+ \big\| \U_{-K} \U_{-K}^H  \a \big\|_2 .
	\end{align*}
	Hence,
	\begin{align*}
	& \big\| \big(\I_M - \tilde{\U} \tilde{\U}^H \big) \a \big\|_2
	- \big\| \big( \I_M - \U_{K} \U_{K}^H \big) \a \big\|_2   \\
	& \leq \: \big\| \big(\I_M - \tilde{\U} \tilde{\U}^H \big) \U_K \U_K^H  \a \big\|_2   \leq \: 2 \sqrt{\tfrac{ M}{p}}  \, \tfrac{ \sigma_{K+1} (\S ) }{ \sigma_{K} (\S ) } \, \big\| \U_K^H \a \big\|_2  ,
	\end{align*}
	where the last inequality holds with probability at least $1-\delta$.

	We leave out $\theta$ in $P_{\textrm{music}} (\theta)$, $\tilde{P}_{\textrm{music}} (\theta)$, and $\a (\theta)$.
	It follows from \eqref{eq:p_music} and \eqref{eq:theory:music_nys} that
	\begin{align*}
	\sqrt{\tfrac{P_{\textrm{music}}}{\tilde{P}_{\textrm{music}}}}
	\:& = \: \tfrac{ \big\| ( \I - \tilde{\U} \tilde{\U}^H ) \a \big\|_2   }{ \big\| ( \I - \U_K \U_K^H ) \a \big\|_2 } \\
	& \leq \: 1 + 2 \sqrt{\tfrac{ M}{p}}  \, \tfrac{ \sigma_{K+1} (\S ) }{ \sigma_{K} (\S ) } \, \tfrac{ \big\|  \U_K^H \a \big\|_2   }{ \big\| ( \I - \U_K \U_K^H ) \a \big\|_2 } .
	\end{align*}

   Furthermore, when the number of element is sufficiently large, for an variant pseudo-spectrum spectrum $||\U_K^H \a||_2$ and $\big\| ( \I - \U_K \U_K^H ) \a \big\|_2$, the following relations always hold:
    \begin{equation} \label{eq:inequality_spetrum_1}
	||\U_K^H \a ||_2^2 \leq ||\U_K^H \a(\theta_k) ||_2^2 \rightarrow M,
    \end{equation}
    and
    \begin{equation} \label{eq:inequality_spetrum_2}
	\big\| ( \I - \U_K \U_K^H ) \a \big\|_2 \geq \big\| ( \I - \U_K \U_K^H ) \a(\theta_k) \big\|_2 \rightarrow 1,
    \end{equation}
    where $\theta_k$ denotes the $k$-th AoA of target point.
    As a result, we further have:
    \begin{equation} \label{eq:inequality_spetrum_ratio}
	\tfrac{ \big\|  \U_K^H \a \big\|_2    }{ \big\| ( \I - \U_K \U_K^H ) \a \big\|_2  } \leq  \tfrac{||\U_K^H \a ||_2} { || \a ||_2} \leq \sqrt{ M}.
    \end{equation}

  On this basis, we can finally obtain the Theorem 1:
  \begin{align*}
	\sqrt{\tfrac{P_{\textrm{music}}}{\tilde{P}_{\textrm{music}}}}
	  \leq \: 1 + 2 \sqrt{\tfrac{ M^2}{p}}  \tfrac{ \sigma_{K+1} (\S ) }{ \sigma_{K} (\S ) } .
	\end{align*}
\end{proof}

\subsection{Proof of Theorem~\ref{thm:power}}

\begin{lemma} \label{lem:power1}
	Let the notation be defined in Section~\ref{sec:theory:power}.
	Let $\Pii$ be any matrix satisfying that $\rk (\U_K^H \Pii) \geq K$.
	Then, for all $\a$,
	\begin{align*}
	& \big\| \U_K \U_K^H \a - \tilde{\U} \tilde{\U}^H  \U_K \U_K^H \a  \big\|_2   \\
	& \leq \: \Big\| (\S - \S_K)^{t+1} \Pii \big( \U_K^H \Pii \big)^\dag \Si_K^{-t-1} \U_K^H \a   \Big\|_2 .
	\end{align*}
\end{lemma}

\begin{proof}
Since $\tilde{\S} = \C \W^\dag \C^H$, it follows from \cite[Lemma 12]{wang2016spsd} that
\begin{equation*}
\tilde{\S} \tilde{\S}^\dag \C
\: = \: \C .
\end{equation*}
Since $\tilde{\U}$ is an orthonormal bases of $\tilde{\S}$,
\begin{equation*}
\tilde{\U} \tilde{\U}^H \C
\: = \tilde{\S} \tilde{\S}^\dag \C
\: = \: \C .
\end{equation*}
By the definition $\C = \S \V$ and that $\V$ is the orthonormal basis of $\S^t \Pii$,
we have that $\C$ and $\S^{t+1} \Pii$ have the same column space.
\begin{equation}\label{eq:lem:power:1}
\tilde{\U} \tilde{\U}^H \S^{t+1} \Pii
\: = \: \S^{t+1} \Pii .
\end{equation}
In the same way as the proof of Lemma~\ref{lem:nys1},
we can prove that for all $\tilde{\x}$:
\begin{align}\label{eq:lem:power:2}
& \big\| \U_K \U_K^H \a - \tilde{\U} \tilde{\U}^H  \U_K \U_K^H \a  \big\|_2
\leq  \big\| \U_K \U_K^H \a - \tilde{\U} \tilde{\x} \big\|_2 .
\end{align}
We artificially construct
\begin{equation*}
\tilde{\x}
\: = \: \tilde{\U}^H \S^{t+1} \Pii \big( \U_K^H \Pii \big)^\dag \Si_K^{-t-1} \U_K^H \a .
\end{equation*}
It follows from \eqref{eq:lem:power:1} that
\begin{align}\label{eq:lem:power:3}
& \tilde{\U} \tilde{\x}
\: = \: \tilde{\U}  \tilde{\U}^H \S^{t+1} \Pii \big( \U_K^H \Pii \big)^\dag \Si_K^{-t-1} \U_K^H \a \nonumber \\
& = \: \big( \S^{t+1}_K + \S^{t+1} - \S^{t+1}_K \big) \Pii \big( \U_K^H \Pii \big)^\dag \Si_K^{-t-1} \U_K^H \a .
\end{align}
In the same way as the proof of Lemma~\ref{lem:nys1}, we can use $\rk (\U_K^H \Pii) \geq K$ to show that
\begin{align*}
&  \S^{t+1}_K  \Pii \big( \U_K^H \Pii \big)^\dag \Si_K^{-t-1} \U_K^H \a \\
& = \: \U_K \Si^{t+1}_K \U_K^H \Pii \big( \U_K^H \Pii \big)^\dag \Si_K^{-t-1} \U_K^H \a  \\
& = \: \U_K \Si^{t+1}_K  \Si_K^{-t-1} \U_K^H \a
\: = \: \U_K \U_K^H \a .
\end{align*}
It follows from \eqref{eq:lem:power:3} that
\begin{align*}
\tilde{\U} \tilde{\x}
\: = \: \U_K \U_K^H \a
+ \:
\big( \S^{t+1} - \S^{t+1}_K \big) \Pii \big( \U_K^H \Pii \big)^\dag \Si_K^{-t-1} \U_K^H \a . \nonumber
\end{align*}
It follows from \eqref{eq:lem:power:2} that
\begin{align*}
& \big\| \U_K \U_K^H \a - \tilde{\U} \tilde{\U}^H  \U_K \U_K^H \a  \big\|_2   \\
& \leq \: \Big\| \big( \S^{t+1} - \S^{t+1}_K \big) \Pii \big( \U_K^H \Pii \big)^\dag \Si_K^{-t-1} \U_K^H \a \Big\|_2  ,
\end{align*}
by which the lemma follows.
\end{proof}

\noindent
{\bf Complete the proof of Theorem~\ref{thm:power}:}
\begin{proof}
	Lemma~\ref{lem:gauss1} shows that the spectral norm of the $M\times K$ standard Gaussian matrix $\Pii$ satisfies
	\begin{equation*}
	\big\| \Pii \big\|_2
	\: \leq \: \sqrt{M } + \sqrt{K} + \OM (1)  , \qquad \textrm{almost surely}.
	\end{equation*}
	Since $\U_K$ has orthonormal columns, the $K\times K$ matrix $\U_K^H \Pii$ is a standard Gaussian matrix.
	It follows from Lemma~\ref{lem:gauss2} that
	\begin{equation*}
	\sigma_{K}^{-1} ( \U_K^H \Pii )
	\: \leq \: \tfrac{\sqrt{ K} } {\delta }
	\end{equation*}
	with probability at least $1-\delta - o(1) $.
	It follows from Lemma~\ref{lem:power1} that
	\begin{align*}
	& \big\| \U_K \U_K^H \a - \tilde{\U} \tilde{\U}^H  \U_K \U_K^H \a  \big\|_2   \\
	& \leq \: \Big\| \big( \S^{t+1} - \S^{t+1}_K \big) \Pii \big( \U_K^H \Pii \big)^\dag \Si_K^{-t-1} \U_K^H \a \Big\|_2 \\
	& \leq \: \big( \tfrac{ \sigma_{K+1} (\S ) }{ \sigma_{K} (\S ) } \big)^{t+1} \,
	\big\| \Pii \big\|_2 \, \big\| ( \U_K^H \Pii )^\dag \big\|_2 \,
	\big\| \U_K^H \a \big\|_2 \\
	& \leq \: \big( \tfrac{ \sigma_{K+1} (\S ) }{ \sigma_{K} (\S ) } \big)^{t+1} \,
	\frac{\sqrt{MK} + K + \OM ( \sqrt{K} )}{\delta}
	\big\| \U_K^H \a \big\|_2 .
	\end{align*}
	The proof follows from the proof of Theorem~\ref{thm:nys} that
	\begin{align*}
	\sqrt{\tfrac{P_{\textrm{music}}}{\tilde{P}_{\textrm{music}}}}
	\:  = \: \tfrac{ \big\| ( \I - \tilde{\U} \tilde{\U}^H ) \a \big\|_2   }{ \big\| ( \I - \U_K \U_K^H ) \a \big\|_2 } \leq \: 1 +  \tfrac{ \big\| \big(\I_M - \tilde{\U} \tilde{\U}^H \big) \U_K \U_K^H  \a \big\|_2  }{ \big\| ( \I - \U_K \U_K^H ) \a \big\|_2 } .
	\end{align*}
	Thus, it holds with probability at least $1-\delta - o(1) $ that
	\begin{align*}
	& \sqrt{\tfrac{P_{\textrm{music}}}{\tilde{P}_{\textrm{music}}}}
	\: \leq \: 1 +
	\big( \tfrac{ \sigma_{K+1} (\S ) }{ \sigma_{K} (\S ) } \big)^{t+1} \,
	\tfrac{ \sqrt{MK} \big( 1+o(1) \big)}{\delta }
	\tfrac{ \big\| \U_K^H \a \big\|_2  }{ \big\| ( \I - \U_K \U_K^H ) \a \big\|_2 } ,
	\end{align*}
	by incorporating eq. \eqref{eq:inequality_spetrum_ratio}, then the theorem can be proved, i.e.
    \begin{align*}
	\sqrt{\tfrac{P_{\textrm{music}}}{\tilde{P}_{\textrm{music}}}}
	\: & \leq \: 1 +
	\big( \tfrac{ \sigma_{K+1} (\S ) }{ \sigma_{K} (\S ) } \big)^{t+1} \,
	\tfrac{ \sqrt{MK} \big( 1+o(1) \big)}{\delta }
	 \sqrt{M}, \\
    & =  1 +
		\tfrac{ \sqrt{M^2K}}{\delta } \big( \tfrac{ \sigma_{K+1} (\S ) }{ \sigma_{K} (\S ) } \big)^{t+1} .
	\end{align*}
\end{proof}

\noindent
{\bf Proof of Theorem~\ref{thm:power2}:}
\begin{proof}
	It can be shown that for all $\tilde{\X}$:
	\begin{align*}
	& \big\| \U_K \U_K^H- \tilde{\U} \tilde{\U}^H  \U_K \U_K^H  \big\|_2
	\leq  \big\| \U_K \U_K^H  - \tilde{\U} \tilde{\X} \big\|_2 .
	\end{align*}
	We artifically construct
	\begin{equation*}
	\tilde{\X}
	\: = \: \tilde{\U}^H \S^{t+1} \Pii \big( \U_K^H \Pii \big)^\dag \Si_K^{-t-1} \U_K^H .
	\end{equation*}
	In the same way as the proof of Lemma~\ref{lem:power1}, we can show that
	\begin{align*}
	\tilde{\U} \tilde{\X}
	\: = \: \U_K \U_K^H
	+ \:
	\big( \S^{t+1} - \S^{t+1}_K \big) \Pii \big( \U_K^H \Pii \big)^\dag \Si_K^{-t-1} \U_K^H  . \nonumber
	\end{align*}
	Thus
	\begin{align*}
	& \big\| \U_K \U_K^H- \tilde{\U} \tilde{\U}^H  \U_K \U_K^H  \big\|_2  \\
	& \leq \:   \big\| \big( \S^{t+1} - \S^{t+1}_K \big) \Pii \big( \U_K^H \Pii \big)^\dag \Si_K^{-t-1} \U_K^H  \big\|_2 \\
	& \leq \:   \big( \tfrac{ \sigma_{K+1} (\S ) }{ \sigma_{K} (\S ) } \big)^{t+1} \,
	\big\| \Pii \big\|_2
	\, \big\| ( \U_K^H \Pii )^\dag \big\|_2 .
	\end{align*}
	It follows from the proof of Theorem~\ref{thm:power} that
	\begin{align*}
	& \big\| \U_K \U_K^H - \tilde{\U} \tilde{\U}^H  \U_K \U_K^H   \big\|_2   \\
	& \leq \: \big[ 1 + o (1) \big] \, \tfrac{ \sqrt{M K} }{\delta} \,
	\big( \tfrac{ \sigma_{K+1} (\S ) }{ \sigma_{K} (\S ) } \big)^{t+1}
	\end{align*}
	holds with probability at least $1-\delta$.
	Since $\U_K$ and $\tilde{\U}$ are both $M\times K$, \cite[Eqn 2.54]{arbenz2012lecture} shows that
	\begin{equation*}
	\big\| ( \I_M - \tilde{\U} \tilde{\U}^H  )\U_K \big\|_2
	\: = \: \big\| ( \I_M - \U_K \U_K^H   ) \tilde{\U}  \big\|_2 .
	\end{equation*}
	It follows that with probability at least $1-\delta$,
	\begin{align} \label{eq:thm:power2:2}
	&\big\| ( \I_M - \U_K \U_K^H   ) \tilde{\U}  \big\|_2
	\: = \:\big\| ( \I_M - \tilde{\U} \tilde{\U}^H  ) \U_K \big\|_2 \nonumber \\
	& \leq  \: \big[ 1 + o (1) \big]  \, \tfrac{ \sqrt{M K} }{\delta} \,
	\big( \tfrac{ \sigma_{K+1} (\S ) }{ \sigma_{K} (\S ) } \big)^{t+1} .
	\end{align}
	We have that
	\begin{align}\label{eq:thm:power2:3}
	& \big\| ( \I_M - \U_K \U_K^H   ) \a \big\|_2 \nonumber \\
	& \leq \: \big\| ( \I_M - \U_K \U_K^H   ) \tilde{\U} \tilde{\U}^H  \a \big\|_2 \nonumber \\
	& \quad + \big\| ( \I_M - \U_K \U_K^H   )  ( \I_M - \tilde{\U} \tilde{\U}^H  ) \a \big\|_2 \nonumber \\
	& \leq \: \big\| ( \I_M - \U_K \U_K^H   ) \tilde{\U} \tilde{\U}^H  \a \big\|_2
	+  \big\|   ( \I_M - \tilde{\U} \tilde{\U}^H  ) \a \big\|_2 .
	\end{align}
	It follows from \eqref{eq:thm:power2:2} and \eqref{eq:thm:power2:3} that
	\begin{align*}
	& \big\| ( \I_M - \U_K \U_K^H   ) \a \big\|_2
	-  \big\|   ( \I_M - \tilde{\U} \tilde{\U}^H  ) \a \big\|_2  \\
	& \leq \:  \big\| ( \I_M - \U_K \U_K^H   ) \tilde{\U} \big\|_2 \, \big\| \a \big\|_2 \\
	& \leq \: \big[ 1 + o (1) \big] \, \tfrac{ \sqrt{M K} }{\delta} \,
	\big( \tfrac{ \sigma_{K+1} (\S ) }{ \sigma_{K} (\S ) } \big)^{t+1}  \big\| \a \big\|_2 ,
	\end{align*}
	where the latter inequality holds with probability at least $1-\delta$.
	Equivalently,
	\begin{align*}
	& \big\|   ( \I_M - \tilde{\U} \tilde{\U}^H  ) \a \big\|_2
	- \big\| ( \I_M - \U_K \U_K^H   ) \a \big\|_2 \\
	& \geq \:   - \big[ 1 + o (1) \big] \,  \tfrac{ \sqrt{M K} }{\delta} \,
	\big( \tfrac{ \sigma_{K+1} (\S ) }{ \sigma_{K} (\S ) } \big)^{t+1}  \big\| \a \big\|_2 ,
	\end{align*}
	Thus, with probability at least $1-\delta$,
	\begin{align*}
	& \sqrt{\tfrac{P_{\textrm{music}}}{\tilde{P}_{\textrm{music}}}}
	\: = \: \tfrac{ \big\| ( \I - \tilde{\U} \tilde{\U}^H ) \a \big\|_2   }{ \big\| ( \I - \U_K \U_K^H ) \a \big\|_2 } \\
	& \geq \: 1 - \big[ 1 + o (1) \big] \, \tfrac{ \sqrt{M K} }{\delta} \,
	\big( \tfrac{ \sigma_{K+1} (\S ) }{ \sigma_{K} (\S ) } \big)^{t+1} \,
	\tfrac{ \| \a \|_2   }{ \| ( \I - \U_K \U_K^H ) \a \|_2 }.
	\end{align*}
    By further applying the inequality in eq. \eqref{eq:inequality_spetrum_2}, we have
    \begin{align*}	
	\tfrac{ \| \a \|_2   }{ \| ( \I - \U_K \U_K^H ) \a \|_2  } \leq \sqrt{M} .
	\end{align*}
	by which the theorem can be proved, i.e.
    \begin{align*}
	\sqrt{\tfrac{P_{\textrm{music}}}{\tilde{P}_{\textrm{music}}}}
	 \: = \: \tfrac{ \big\| ( \I - \tilde{\U} \tilde{\U}^H ) \a \big\|_2   }{ \big\| ( \I - \U_K \U_K^H ) \a \big\|_2 }  \geq \: 1 -  \tfrac{ \sqrt{M^2 K} }{\delta} \,
	\big( \tfrac{ \sigma_{K+1} (\S ) }{ \sigma_{K} (\S ) } \big)^{t+1} .
	\end{align*}
\end{proof}

\bibliography{bib/scibib,bib/scibib2}

\begin{thebibliography}{10}

\bibitem{jones2001keeping}
W.~D. Jones, ``Keeping cars from crashing,'' {\em IEEE Spectrum}, vol.~38,
  no.~9, pp.~40--45, 2001.

\bibitem{Guizzo2011Google}
D.~Guizzo, ``How google's self-driving car works,'' {\em IEEE Spectrum Online},
  vol.~18, 2011.

\bibitem{Blasch2006Unmanned}
E.~P. Blasch, A.~Lakhotia, and G.~Seetharaman, ``Unmanned vehicles come of age:
  The darpa grand challenge,'' {\em Computer}, vol.~39, no.~12, pp.~26--29,
  2006.

\bibitem{Alonzo2004Toward}
K.~Alonzo, A.~Stentz, O.~Amidi, M.~Bode, D.~Bradley, A.~Diazcalderon,
  M.~Happold, H.~Herman, R.~Mandelbaum, and T.~Pilarski, ``Toward reliable off
  road autonomous vehicles operating in challenging environments,'' in {\em
  Proc of International Symposium on Exprimental Robotics}, 2004.

\bibitem{murad2013requirements}
M.~Murad, I.~Bilik, M.~Friesen, J.~Nickolaou, J.~Salinger, K.~Geary, and J.~S.
  Colburn, ``Requirements for next generation automotive radars,'' {\em in
  Proc. of 2013 IEEE Radar Conference}, pp.~1--6, 2013.

\bibitem{waldschmidt2014future}
C.~Waldschmidt and H.~H. Meinel, ``Future trends and directions in radar
  concerning the application for autonomous driving,'' {\em in Proc. of 2014
  European Microwave Conference}, pp.~1719--1722, 2014.

\bibitem{hasch2012millimeter}
J.~Hasch, E.~Topak, R.~Schnabel, T.~Zwick, R.~Weigel, and C.~Waldschmidt,
  ``Millimeter-wave technology for automotive radar sensors in the 77 ghz
  frequency band,'' {\em IEEE Transactions on Microwave Theory and Techniques},
  vol.~60, no.~3, pp.~845--860, 2012.

\bibitem{li2007mimo}
J.~Li and P.~Stoica, ``Mimo radar with colocated antennas,'' {\em IEEE Signal
  Processing Magazine}, vol.~24, no.~5, pp.~106--114, 2007.

\bibitem{bilik2018automotive}
I.~Bilik, S.~Villeval, D.~Brodeski, H.~Ringel, O.~Longman, P.~Goswami, C.~Y.~B.
  Kumar, S.~Rao, P.~Swami, A.~Jain, {\em et~al.}, ``Automotive multi-mode
  cascaded radar data processing embedded system,'' in {\em Proceedings of 2018
  IEEE Radar Conference}, pp.~1--6, 2018.

\bibitem{2019The}
I.~Bilik, O.~Longman, S.~Villeval, and J.~Tabrikian, ``{The Rise of Radar for
  Autonomous Vehicles: Signal Processing Solutions and Future Research
  Directions},'' {\em IEEE Signal Processing Magazine}, vol.~36, no.~5,
  pp.~20--31, 2019.

\bibitem{patole2017automotive}
S.~Patole, M.~Torlak, D.~Wang, and M.~Ali, ``Automotive radars: A review of
  signal processing techniques,'' {\em IEEE Signal Processing Magazine},
  vol.~34, no.~2, pp.~22--35, 2017.

\bibitem{2019High}
G.~Hakobyan and B.~Yang, ``{High-Performance Automotive Radar: A Review of
  Signal Processing Algorithms and Modulation Schemes},'' {\em IEEE Signal
  Processing Magazine}, vol.~36, no.~5, pp.~32--44, 2019.

\bibitem{stove1992linear}
A.~G. Stove, ``Linear {FMCW} radar techniques,'' {\em Proceedings of the IEEE},
  1992.

\bibitem{Lin2011A}
C.~H. Lin, Y.~S. Wu, Y.~L. Yeh, S.~H. Weng, G.~Y. Chen, C.~H. Shen, and H.~Y.
  Chang, ``{A 24-GHz highly integrated transceiver in 0.5-$\mu$m E/D-PHEMT
  process for FMCW automotive radar applications},'' in {\em Microwave
  Conference Proceedings}, pp.~512--515, 2011.

\bibitem{Babur2013Nearly}
G.~Babur, O.~A. Krasnov, A.~Yarovoy, and P.~Aubry, ``{Nearly Orthogonal
  Waveforms for MIMO FMCW Radar},'' {\em IEEE Transactions on Aerospace \&
  Electronic Systems}, vol.~49, no.~3, pp.~1426--1437, 2013.

\bibitem{Kok2005Signal}
D.~Kok and J.~S. Fu, ``Signal processing for automotive radar,'' in {\em Radar
  Conference, 2005 IEEE International}, pp.~842--846, 2005.

\bibitem{2020MIMO}
S.~Sun, A.~P. Petropulu, and H.~V. Poor, ``{MIMO Radar for Advanced
  Driver-Assistance Systems and Autonomous Driving: Advantages and
  Challenges},'' {\em IEEE Signal Processing Magazine}, vol.~37, no.~4,
  pp.~98--117, 2020.

\bibitem{Wenger2005Automotive}
J.~Wenger, ``Automotive radar - status and perspectives,'' in {\em Compound
  Semiconductor Integrated Circuit Symposium, 2005 (CSIC '05)}, p.~4 pp., 2005.

\bibitem{schmidt1986multiple}
R.~O. Schmidt, ``Multiple emitter location and signal parameter estimation,''
  {\em IEEE Transactions on Antennas and Propagation}, vol.~34, no.~3,
  pp.~276--280, 1986.

\bibitem{roy1989esprit}
R.~H. Roy and T.~Kailath, ``{ESPRIT}-estimation of signal parameters via
  rotational invariance techniques,'' {\em IEEE Transactions on Acoustics,
  Speech, and Signal Processing}, vol.~37, no.~7, pp.~984--995, 1989.

\bibitem{krim1996two}
H.~Krim and M.~Viberg, ``Two decades of array signal processing research: the
  parametric approach,'' {\em IEEE Signal Processing Magazine}, vol.~13, no.~4,
  pp.~67--94, 1996.

\bibitem{oh2015low}
D.~Oh and J.~Lee, ``Low-complexity range-azimuth fmcw radar sensor using joint
  angle and delay estimation without svd and evd,'' {\em IEEE Sensors Journal},
  vol.~15, no.~9, pp.~4799--4811, 2015.

\bibitem{zoltowski1993beamspace}
M.~D. Zoltowski, G.~M. Kautz, and S.~D. Silverstein, ``{Beamspace
  Root-MUSIC},'' {\em IEEE Transactions on Signal Processing}, vol.~41, no.~1,
  pp.~344--364, 1993.

\bibitem{simon1984the}
H.~D. Simon, ``The lanczos algorithm with partial reorthogonalization,'' {\em
  Mathematics of Computation}, vol.~42, no.~165, pp.~115--142, 1984.

\bibitem{Mathews1994Eigenstructure}
Mathews, P.~C., Zoltowski, and M.~D~., ``{Eigenstructure techniques for 2-D
  angle estimation with uniform circular arrays},'' {\em Signal Processing,
  IEEE Transactions on}, vol.~42, no.~9, pp.~2395--2407, 1994.

\bibitem{2018Lin}
Z.~Lin, T.~Lv, and P.~T. Mathiopoulos, ``{3-D Indoor Positioning for
  Millimeter-Wave Massive MIMO Systems},'' {\em IEEE Transactions on
  Communications}, vol.~66, no.~6, pp.~2472--2486, 2018.

\bibitem{2013Spatial}
M.~Rossi, A.~M. Haimovich, and Y.~C. Eldar, ``{Spatial Compressive Sensing for
  MIMO Radar},'' {\em IEEE Transactions on Signal Processing}, vol.~62, no.~2,
  pp.~419--430, 2013.

\bibitem{Stoeckle2015DoA}
Stoeckle, Christoph, Munir, Jawad, Mezghani, Nossek, Munich, and Germany),
  ``{DoA Estimation Performance and Computational Complexity of Subspace- and
  Compressed Sensing-based Methods},'' in {\em in Proc. of 19th International
  ITG Workshop on Smart Antennas,}, 2015.

\bibitem{2015Distributed}
B.~Li and A.~P. Petropulu, ``{Distributed MIMO radar based on sparse sensing:
  Analysis and efficient implementation},'' {\em IEEE Transactions on Aerospace
  \& Electronic Systems}, vol.~51, no.~4, pp.~3055--3070, 2015.

\bibitem{Benesty2005A}
J.~Benesty, J.~Chen, and Y.~Huang, ``A generalized {MVDR} spectrum,'' {\em IEEE
  Signal Processing Letters}, vol.~12, no.~12, pp.~827--830, 2005.

\bibitem{marcos1995the}
S.~Marcos, A.~Marsal, and M.~Benidir, ``The propagator method for source
  bearing estimation,'' {\em Signal Processing}, vol.~42, no.~2, pp.~121--138,
  1995.

\bibitem{Marcos1994Performances}
S.~Marcos, A.~Marsal, and M.~Benidir, ``Performances analysis of the propagator
  method for source bearing estimation,'' in {\em Acoustics, Speech, \& Signal
  Processing, on IEEE International Conference}, 1994.

\bibitem{Tayem2005L}
N.~Tayem and H.~M. Kwon, ``L-shape 2-dimensional arrival angle estimation with
  propagator method,'' {\em IEEE Trans Antennas \& Propagat}, vol.~1, no.~5,
  pp.~1622--1630, 2005.

\bibitem{engels2017advances}
F.~Engels, P.~Heidenreich, A.~M. Zoubir, F.~K. Jondral, and M.~Wintermantel,
  ``Advances in automotive radar: A framework on computationally efficient
  high-resolution frequency estimation,'' {\em IEEE Signal Processing
  Magazine}, vol.~34, no.~2, pp.~36--46, 2017.

\bibitem{Garcia2018TI}
M.~Y. Keegan~Garcia and A.~Purkovic, ``{Robust traffic and intersection
  monitoring using millimeter wave sensors},'' tech. rep., Texas Instruments
  Incorporated (TI), May 2018.

\bibitem{Ramasubramanian2017TI}
K.~Ramasubramanian, ``{Using a complex-baseband architecture in {FMCW} radar
  systems},'' tech. rep., Texas Instruments Incorporated (TI), May 2017.

\bibitem{2012Efficient}
J.~He, M.~N.~S. Swamy, and M.~O. Ahmad, ``{Efficient Application of MUSIC
  Algorithm Under the Coexistence of Far-Field and Near-Field Sources},'' {\em
  IEEE Transactions on Signal Processing}, vol.~60, no.~4, pp.~2066--2070,
  2012.

\bibitem{Li2020Fast}
B.~Li, S.~Wang, J.~Zhang, X.~Cao, and C.~Zhao, ``{Fast Randomized-MUSIC for
  mm-Wave Massive MIMO Radars},'' {\em IEEE Transactions on Vehicular
  Technology}, vol.~70, no.~2, pp.~1952--1956.

\bibitem{stoica1989music}
P.~Stoica and N.~Arye, ``{MUSIC, maximum likelihood, and Cramer-Rao bound},''
  {\em IEEE Transactions on Acoustics, Speech, and Signal Processing}, vol.~37,
  no.~5, pp.~720--741, 1989.

\bibitem{fortunati2019scaling}
S.~Fortunati, L.~Sanguinetti, M.~S. Greco, and F.~Gini, ``{Scaling up MIMO
  radar for target detection},'' in {\em {2019 IEEE International Conference on
  Acoustics, Speech and Signal Processing (ICASSP 2019)}}, 2019.

\bibitem{vallet2015performance}
P.~Vallet, X.~Mestre, and P.~Loubaton, ``{Performance Analysis of an Improved
  MUSIC DoA Estimator},'' {\em IEEE Transactions on Signal Processing},
  vol.~63, no.~23, pp.~6407--6422, 2015.

\bibitem{Liao2014MUSIC}
W.~Liao and A.~Fannjiang, ``{MUSIC for Single-Snapshot Spectral Estimation:
  Stability and Super-resolution},'' {\em Applied \& Computational Harmonic
  Analysis}, vol.~40, no.~1, pp.~33--67, 2014.

\bibitem{Li2021Fast}
B.~Li, S.~Wang, Z.~Feng, J.~Zhang, X.~Cao, and C.~Zhao, ``{Fast Pseudo-spectrum
  Estimation for Automotive Massive MIMO Radar},'' {\em IEEE Internet of Things
  Journal}, vol.~8, no.~20, pp.~15303--15316, 2021.

\bibitem{2020Transform}
H.~Cao, Q.~Liu, and Y.~Wu, ``{Transform Domain: Design of Closed-Form Joint 2-D
  DOA Estimation Based on QR Decomposition},'' {\em Circuits Systems and Signal
  Processing}, pp.~1--12, 2020.

\bibitem{2019DOA}
Q.~Liu, Y.~Gu, and H.~C. So, ``{DOA Estimation in Impulsive Noise via Low-Rank
  Matrix Approximation and Weakly Convex Optimization},'' {\em IEEE
  Transactions on Aerospace and Electronic Systems}, vol.~55, no.~6, 2019.

\bibitem{2020Target}
Q.~Liu, J.~Xu, Z.~Ding, and H.~C. So, ``{Target Localization With Jammer
  Removal Using Frequency Diverse Array},'' {\em IEEE Transactions on Vehicular
  Technology}, vol.~69, no.~10, pp.~11685--11696, 2020.

\bibitem{wang2019scalable}
S.~Wang, A.~Gittens, and M.~W. Mahoney, ``Scalable kernel k-means clustering
  with {N}ystrom approximation: Relative-error bounds,'' {\em Journal of
  Machine Learning Research}, vol.~20, no.~12, pp.~1--49, 2019.

\bibitem{drineas2005nystrom}
P.~Drineas and M.~W. Mahoney, ``On the {N}ystr\"{o}m method for approximating a
  gram matrix for improved kernel-based learning,'' {\em Journal of Machine
  Learning Research}, vol.~6, pp.~2153--2175, 2005.

\bibitem{woodruff2014sketching}
D.~P. Woodruff, ``Sketching as a tool for numerical linear algebra,'' {\em
  Foundations and Trends{\textregistered} in Theoretical Computer Science},
  vol.~10, no.~1--2, pp.~1--157, 2014.

\bibitem{2021Random}
B.~Li, P.~Chen, H.~Liu, W.~S. Guo, X.~B. Cao, J.~Z. Du, C.~L. Zhao, and
  J.~Zhang, ``{Random Sketch Learning for Deep Neural Networks in Edge
  Computing},'' {\em Nature Computational Science}, vol.~1, no.~3,
  pp.~221--228, 2021.

\bibitem{Li2020Randomized}
B.~Li, S.~Wang, J.~Zhang, X.~Cao, and C.~Zhao, ``{Randomized Approximate
  Channel Estimator in Massive-MIMO Communication},'' {\em IEEE Communications
  Letters}, vol.~24, no.~10, pp.~2314--2318, 2020.

\bibitem{tropp2017fixed}
J.~A. Tropp, A.~Yurtsever, M.~Udell, and V.~Cevher, ``Fixed-rank approximation
  of a positive-semidefinite matrix from streaming data,'' in {\em Advances in
  Neural Information Processing Systems (NIPS)}, 2017.

\bibitem{gittens2016revisiting}
A.~Gittens and M.~W. Mahoney, ``Revisiting the {N}ystrom method for improved
  large-scale machine learning,'' {\em Journal of Machine Learning Research},
  vol.~17, no.~1, pp.~3977--4041, 2016.

\bibitem{1995Noise}
B.~M. Sadler, G.~B. Giannakis, and S.~Shamsunder, ``Noise subspace techniques
  in non-gaussian noise using cumulants,'' {\em IEEE Transactions on Aerospace
  \& Electronic Systems}, vol.~31, no.~3, pp.~1009--1018, 1995.

\bibitem{candes2009exact}
E.~J. Candes and B.~Recht, ``Exact matrix completion via convex optimization,''
  {\em Foundations of Computational mathematics}, vol.~9, no.~6, p.~717, 2009.

\bibitem{wang2016spsd}
S.~Wang, L.~Luo, and Z.~Zhang, ``{SPSD} matrix approximation vis column
  selection: Theories, algorithms, and extensions,'' {\em Journal of Machine
  Learning Research}, vol.~17, no.~49, pp.~1--49, 2016.

\bibitem{Akaike1974A}
H.~Akaike, ``A new look at statistical model identification,'' {\em IEEE
  Transactions on Automatic Control}, no.~6, pp.~716--723, 1974.

\bibitem{Barron1998The}
A.~Barron, J.~Rissanen, and B.~Yu, ``The minimum description length principle
  in coding and modeling,'' {\em IEEE Transactions on Information Theory},
  vol.~44, no.~6, pp.~2743--2760, 1998.

\bibitem{1995The}
J.~K. Thomas and L.~L. Scharf, ``{The probability of a subspace swap in the
  SVD},'' {\em IEEE Transactions on Signal Processing}, vol.~43, no.~3,
  pp.~730--736, 1995.

\bibitem{vershynin2010introduction}
R.~Vershynin, ``Introduction to the non-asymptotic analysis of random
  matrices,'' {\em arXiv preprint arXiv:1011.3027}, 2010.

\bibitem{rudelson2008littlewood}
M.~Rudelson and R.~Vershynin, ``The {L}ittlewood--{O}fford problem and
  invertibility of random matrices,'' {\em Advances in Mathematics}, vol.~218,
  no.~2, pp.~600--633, 2008.

\bibitem{tao2010random}
T.~Tao and V.~Vu, ``Random matrices: The distribution of the smallest singular
  values,'' {\em Geometric And Functional Analysis}, vol.~20, no.~1,
  pp.~260--297, 2010.

\bibitem{arbenz2012lecture}
P.~Arbenz, ``Lecture notes on solving large scale eigenvalue problems,'' {\em
  D-MATH, EHT Zurich}, vol.~2, 2012.

\end{thebibliography}

\bibliographystyle{ieeetr}

\end{document}